\documentclass[sigconf]{acmart}

\usepackage{xcolor}
\usepackage{xspace}
\usepackage{enumitem}

\usepackage{multirow}

\newcommand{\hideit}[1]{}

\usepackage{algorithm}
\usepackage{algpseudocode}

\begin{document}

\title{ManCAR: Manifold-Constrained Latent Reasoning with Adaptive Test-Time Computation for Sequential Recommendation }

\author{
  Kun Yang$^{1}$, Yuxuan Zhu$^{2}$, Yazhe Chen$^1$, Siyao Zheng$^3$, Bangyang Hong$^{2}$, Kangle Wu$^{2}$, Yabo Ni$^2$, Anxiang Zeng$^2$, Cong Fu$^2$, Hui Li$^1\dagger$
}
\renewcommand{\shortauthors}{Kun Yang et al.}
\thanks{
    $^1$Key Laboratory of Multimedia Trusted Perception and Efficient Computing, Ministry of Education of China, Xiamen University, China. \\
    $^2$Shopee Pte. Ltd. \\
    $^3$School of Informatics, Xiamen University, China. \\
    $\dagger$ Hui Li is the corresponding author. hui@xmu.edu.cn
}

\begin{abstract}

Sequential recommendation increasingly employs latent multi-step reasoning to enhance test-time computation. Despite empirical gains, existing approaches largely drive intermediate reasoning states via target-dominant objectives without imposing explicit feasibility constraints. This results in latent drift, where reasoning trajectories deviate into implausible regions. We argue that effective recommendation reasoning should instead be viewed as navigation on a collaborative manifold rather than free-form latent refinement. To this end, we propose ManCAR (Manifold-Constrained Adaptive Reasoning), a principled framework that grounds reasoning within the topology of a global interaction graph. ManCAR constructs a local intent prior from the collaborative neighborhood of a user's recent actions, represented as a distribution over the item simplex. During training, the model progressively aligns its latent predictive distribution with this prior, forcing the reasoning trajectory to remain within the valid manifold. At test time, reasoning proceeds adaptively until the predictive distribution stabilizes, avoiding over-refinement. We provide a variational interpretation of ManCAR to theoretically validate its drift-prevention and adaptive test-time stopping mechanisms. Experiments on seven benchmarks demonstrate that ManCAR consistently outperforms state-of-the-art baselines, achieving up to a 46.88\% relative improvement w.r.t. NDCG@10. Our code is available at https://github.com/FuCongResearchSquad/ManCAR.

\end{abstract}



\maketitle
\section{Introduction}
\label{sec:intro}

\begin{figure}[t]  
\centering
\includegraphics[width=0.95\columnwidth]{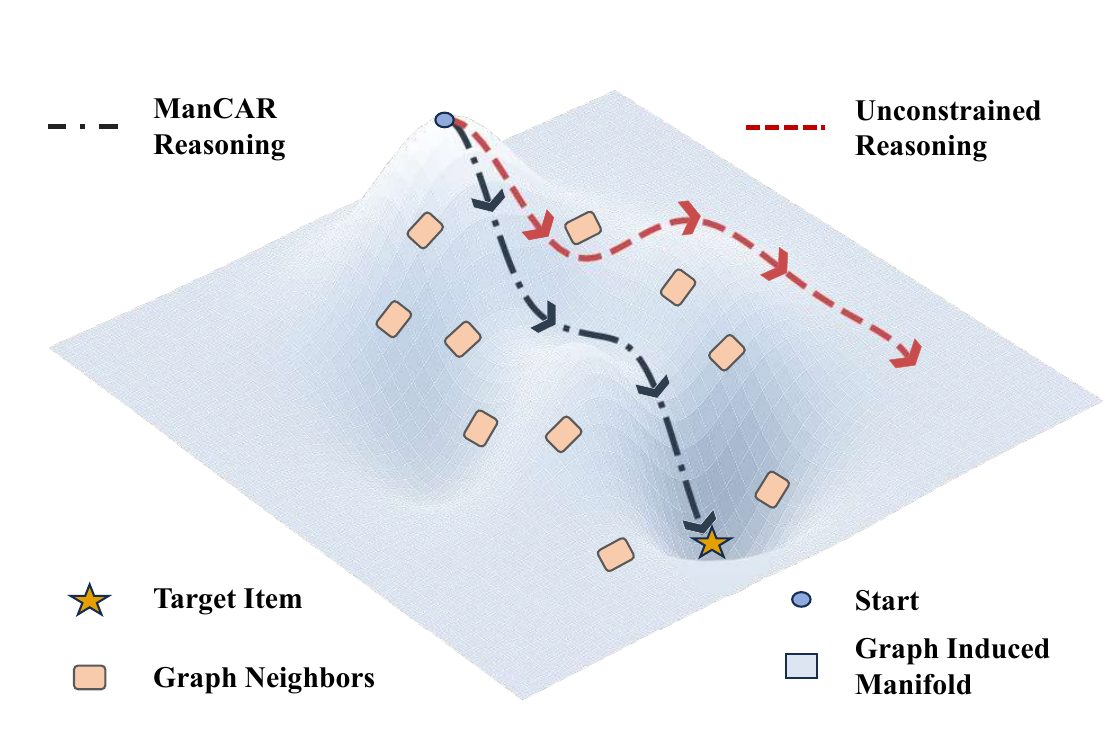}
\vspace{-15pt}
\caption{Illustration of constrained versus unconstrained latent reasoning. Graph-conditioned reasoning trajectories remain within a collaborative manifold defined by neighbor items, enabling stable and directed refinement toward the target. In contrast, unconstrained reasoning may drift outside feasible regions,
leading to inefficient or unstable paths.}
\vspace{-10pt}
\label{fig:motivation}
\end{figure}

Sequential recommendation has been significantly reshaped by the growing adoption of generative modeling paradigms~\cite{DeldjooHMKSRVSK24}.  
Inspired by Large Language Models (LLMs), recent work has begun to explore latent multi-step reasoning in sequential recommendation to extend test-time computation~\cite{abs-2503-22675,abs-2509-18091,abs-2505-16865,abs-2601-03153}. In LLMs, such reasoning is commonly realized through the recursive depth paradigm, where chain-of-thought (CoT) tokens are replaced by \textbf{un-decoded} latent states produced by the model~\cite{abs-2412-06769,ShenYZHDH25,abs-2510-15522,abs-2509-02350}. These latent states are iteratively refined for multiple steps using shared model parameters, and only decoded back to the output space at the final step, effectively increasing the model's computational depth without expanding its architecture. This paradigm offers a natural and efficient template for incorporating reasoning into sequential recommendation, without requiring explicit textual representations.

Despite empirical gains, existing methods remain poorly understood. They typically guide latent reasoning using target-dominant objectives~\cite{abs-2503-22675,abs-2601-03153}, such as supervising only the final reasoning state with the target item, or mapping each intermediate state to an item probability distribution and progressively concentrating it toward a target \emph{one-hot} distribution. However, they impose no \emph{explicit} constraints on the evolution of intermediate reasoning states. As a result, the latent reasoning trajectory is largely unconstrained and retains excessive degrees of freedom while ``walking'' through the item space. This often leads to \textbf{latent drift} (Fig.~\ref{fig:motivation}), where intermediate states migrate into regions that are poorly aligned with user preferences. Such drift is particularly detrimental at test time, degrading model robustness and generalization.

From a recommendation perspective, an overlooked but fundamental property of user behavior is that interactions are inherently collaborative rather than independent. This naturally motivates guiding latent reasoning using an item interaction graph, which encodes collective patterns across users. Users with similar preferences tend to interact with similar items, and item transitions exhibit regularities shaped by population-level behaviors. These collaborative signals naturally define a notion of plausibility: given a user's recent interactions, only a subset of items is realistically relevant in the near future. Such assumptions are widely adopted in graph-based recommendation~\cite{WuT0WXT19,YingHCEHL18,WeiRTWSCWYH24,ChangGZHNSJ021,YangHXL22,YangWWHHW24,JuFGLLQQSSXYYZWLZ24,YuY00CN22}, where a user's next interaction is expected to lie within the local neighborhood of their recent interests.

In this paper, we propose ManCAR (\underline{Man}ifold-\underline{C}onstrained \underline{A}daptive \underline{R}easoning), a principled framework that grounds latent reasoning within the topology of a global interaction graph. 
Rather than naively enumerating graph traversal paths as reasoning trajectories, which is computationally expensive and unnecessary for latent reasoning, we leverage the interaction graph as a feasibility constraint on the reasoning process. Specifically, we treat the neighborhood induced by the item graph as a \textbf{manifold constraint}, restricting latent reasoning trajectories to evolve within collaboratively reachable regions while refining toward the target item. In probabilistic terms, this constraint corresponds to a region on the item probability simplex where items connected to the user's recent actions are assigned substantially higher probability mass than unrelated items. This feasibility view naturally admits a \textbf{variational interpretation of latent reasoning}. Introducing latent reasoning states can be viewed as performing inference over an intermediate intent variable, with the graph-induced neighborhood serving as a structure-aware prior. Latent reasoning can then be formulated using an objective similar to the Evidence Lower Bound (ELBO), which balances target prediction with reasoning feasibility.

Besides, while the manifold constraint defines where latent reasoning can evolve, it leaves open the question of when reasoning should terminate. Since we train reasoning states to traverse collaboratively feasible regions on the item probability simplex toward the target, further refinement becomes uninformative once the item probability distribution produced by the latent state stabilizes. This motivates us to design a \textbf{convergence-based stopping criterion} for MacCAR, allowing test-time computation to terminate adaptively when the model has sufficiently localized the target region. 

Our contributions can be summarized as follows:
\begin{itemize}[leftmargin=*]
    \item We propose ManCAR, a framework guiding latent reasoning by interpreting collaborative neighborhoods in the interaction graph as feasibility constraints on the item probability simplex, mitigating latent drift. ManCAR further enables adaptive test-time computation via a convergence-based stopping criterion.
    
    \item We theoretically establish a variational interpretation of ManCAR, demonstrating how it prevents latent drift and confirm the validity of our adaptive test-time stopping mechanism.
    
    \item Experiments on benchmarks demonstrate that ManCAR consistently improves effectiveness over state-of-the-art baselines, achieving up to a 46.88\% relative improvement w.r.t. NDCG@10.
    
\end{itemize}

\section{Our Proposed ManCAR}
Fig.~\ref{fig:framework} depicts ManCAR illustrated in this section.
We begin by formalizing the problem setting and the used notation in Sec.~\ref{sec:problem}. 
In Sec.~\ref{sec:manifoldreason}, we introduce the core concept of manifold-constrained reasoning, postulating that valid latent reasoning trajectories should be confined to the local neighborhood of the user's recent interactions. 
To operationalize this, Sec.~\ref{sec:variationalobj} derives a variational training objective that treats latent reasoning as approximate inference over an intent variable, regularized by a graph-conditioned teacher prior. 
Sec.~\ref{sec:graphsmoothness} provides a theoretical analysis of this objective, interpreting the KL-divergence regularization as a gradient flow that enforces local graph smoothness. 
Based on this theoretical foundation, Sec.~\ref{sec:implementation} details the practical implementation of ManCAR, including teacher prior construction strategies and the overall loss function. 
Finally, Sec.~\ref{sec:adaptive_reasoning} introduces a training mechanism for scheduling the teacher distribution during training, which naturally enables an adaptive termination criterion at test time.

\begin{figure}[t]  
  \centering
  \includegraphics[width=\columnwidth]{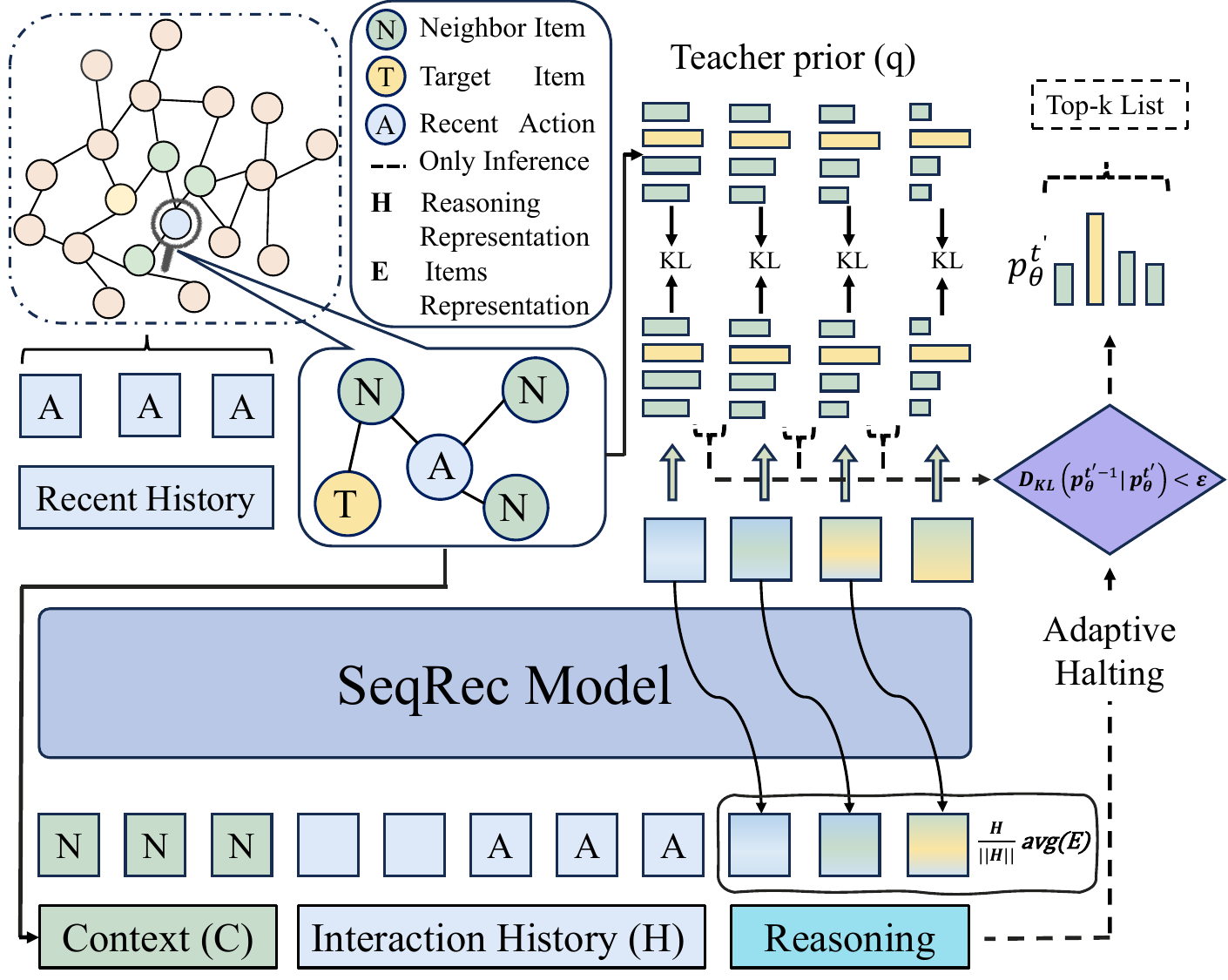}
  \vspace{-20pt}
  \caption{Overview of ManCAR. ManCAR performs multi-step latent reasoning constrained by a graph-induced candidate set. At each step, the reasoning state is regularized toward a scheduled teacher prior defined on collaboratively reachable items, ensuring manifold-consistent refinement. Adaptive test-time termination stops reasoning when the induced item distributions stabilize.}
  \vspace{-10pt}
  \label{fig:framework}
\end{figure}

\subsection{Problem Setting and Notation}
\label{sec:problem}

Let $\mathcal{I}$ denote a finite set of items and $\mathcal{U}$ the set of users. We consider the sequential recommendation setting, where each user $u\in\mathcal{U}$ interacts with items over time. For a given user, the interaction history is denoted by $H = (i_1, i_2, \dots, i_{T-1}), i_t \in \mathcal{I}$,
where items are ordered chronologically. The objective is to predict the next item $i^* = i_T$ conditioned on the observed history $H$.

We model collaborative signals among items using a global \emph{item interaction graph} $\mathcal{G}=(\mathcal{I},\mathcal{E})$, where nodes correspond to items and edges encode co-interaction relationships aggregated across users. An edge $(i,j)\in\mathcal{E}$ indicates that items $i$ and $j$ are frequently co-interacted or consecutively consumed by users. For an item $i\in\mathcal{I}$, we denote its $k$-hop graph neighborhood by $\mathcal{N}(i;\mathcal{G};k)$.

Following standard latent reasoning settings~\cite{abs-2503-22675}, we denote by $\mathbf{h}_t \in \mathbb{R}^d$ the hidden representation produced by the backbone encoder at step $t$, i.e., $\mathbf{h}_t = f(H[:t])$, where $f(\cdot)$ is typically a Transformer-based encoder. We further introduce $\mathbf{r}_{t'} \in \mathbb{R}^d$, $t' \in \{1,\dots,T'\}$, to denote the latent reasoning states generated through iterative refinement. Unless otherwise specified, the initial reasoning state is set to the final encoder state, i.e., $\mathbf{r}_1 = \mathbf{h}_{T-1}$.

\subsection{Manifold-Constrained Latent Reasoning}
\label{sec:manifoldreason}

We now articulate the core conceptual motivation behind ManCAR. 
As discussed in Sec.~\ref{sec:intro}, standard latent reasoning based methods suffer from unconstrained degrees of freedom, leading to latent drift. 
To address this, we propose a geometric perspective: reasoning should not traverse the latent space freely but must be viewed as navigation constrained to a ``collaborative manifold''. 

We adopt a latent multi-step reasoning setting~\cite{abs-2503-22675}. Starting from the final encoder state $\mathbf{h}^{T-1}$, the model generates a sequence of reasoning states $\mathbf{r}_{t'}$ ($t' \in \{1,\dots,T'\}$), each representing an intermediate hypothesis of user intent, and produces the final recommendation by decoding the last state.

A central challenge in this setting is that, without additional knowledge, the evolution of latent reasoning states is weakly constrained and exhibits excessive degrees of freedom, particularly in high-dimensional spaces. To address this issue, we introduce a \emph{graph-conditioned feasibility constraint} that explicitly restricts where latent reasoning can evolve. This design is motivated by a fundamental property of recommender systems: user behavior is inherently collaborative, and given a user's most recent interactions $I_n=(i_{T-n},\dots,i_{T-1})$ where $n$ denotes the window size of the recent interaction, only a limited subset of items is plausibly relevant in the near future~\cite{WuT0WXT19,YingHCEHL18}—namely, those that are $k$-hop reachable on the collaborative item interaction graph $\mathcal{G}$.

From a geometric perspective, this graph-conditioned neighborhood defines a low-dimensional feasible region within the high-dimensional item space. To operationalize this constraint during latent reasoning, we explicitly regulate how reasoning states are translated into item-level beliefs. Specifically, each reasoning state is mapped to an item probability distribution on the item probability simplex, and the graph-conditioned neighborhood restricts valid distributions to those that concentrate probability mass on collaboratively reachable items. This restriction defines a structured subregion of the simplex, which we refer to as the \textbf{collaborative manifold}. Latent reasoning is therefore constrained to evolve along this manifold, rather than freely over the entire simplex as in unconstrained latent refinement (Fig.~\ref{fig:motivation}).

To make the collaborative manifold explicit and tractable, we instantiate it using a finite set of collaboratively reachable items conditioned on the user's most recent interactions. This set defines the feasible support of item probability distributions during latent reasoning. Concretely, we define a finite candidate set:
\begin{equation*}
    \mathcal{C}(I_n,\mathcal{G},k) \subseteq \{I_n\} \cup \mathcal{N}(I_n;\mathcal{G};k).
\end{equation*}
Unless otherwise specified, we use $\mathcal{C}(k)$ as shorthand for $\mathcal{C}(I_n, \mathcal{G}, k)$, since the candidate set is always constructed from the most recent items $I_n$ on the interaction graph $\mathcal{G}$ in this paper. In the next subsection, we derive a training objective from a variational interpretation that leverages this graph-conditioned feasible set to regularize latent reasoning.

\subsection{Variational Training Objective}
\label{sec:variationalobj}

While the manifold constraint provides a strong geometric intuition, optimizing it effectively requires a rigorous mathematical formulation. 
In this subsection, we translate the conceptual constraint from Sec.~\ref{sec:manifoldreason} into \textbf{probabilistic variational inference}. We derive an Evidence Lower Bound (ELBO)-like objective, treating latent reasoning as approximate inference over an intermediate intent variable. This derivation introduces the ``Teacher Prior'' distribution induced from the interaction graph, guiding the implementation of the loss function (Sec.~\ref{sec:implementation}) necessary to train the model.

\vspace{5pt}
\noindent\textbf{Latent Variable Formulation.}
Given a user history $H$, we introduce a discrete latent variable $c \in \mathcal{C}(k)$ representing an intermediate intent prototype that mediates the prediction of the next item. The conditional likelihood of the target item can be written as
\begin{equation*}
p_\theta(i^* \mid H) = \sum_{c \in \mathcal{C}(k)} p_\theta(c \mid H)\, p_\theta(i^* \mid c, H),
\end{equation*}
where $p_\theta(c \mid H)$ is a history-conditioned intent distribution parameterized by $\theta$, and $p_\theta(i^* \mid c, H)$ models the likelihood of the target given the inferred intent. We restrict $c$ to lie in the graph-induced candidate set, reflecting the assumption that user intent at each reasoning step is best characterized by collaboratively reachable items. This design anchors latent intent to observable interaction patterns and aligns with the collaborative manifold constraint.

\vspace{5pt}
\noindent\textbf{Graph-Conditioned Variational Prior.}
Based on the above design, we introduce a graph-conditioned teacher distribution $q(c \mid I_n, \mathcal{G})$, defined over the candidate set $\mathcal{C}(k)$. It encodes prior knowledge about plausible intents reachable from the user's most recent interaction, and is constructed independently of the model parameters $\theta$. Intuitively, $q$ assigns higher probability mass to items that are strongly connected to $I_n$ in the interaction graph.

\vspace{5pt}
\noindent\textbf{ELBO-like Objective.}
For any choice of $q(c \mid I_n, \mathcal{G})$, the log-likelihood admits the following lower bound:
\begin{equation}
\label{equation:elbo-obj}
\begin{aligned}
\log p_\theta(i^* \mid H)
\;\ge\;
&\mathbb{E}_{q(c \mid I_n, \mathcal{G})}
\big[ \log p_\theta(i^* \mid c, H) \big] \\
&-
D_{\mathrm{KL}}\!\left(
q(c \mid I_n, \mathcal{G})
\;\|\;
p_\theta(c \mid H)
\right),
\end{aligned}
\end{equation}
where $D_{KL}(\cdot\;\|\;\cdot)$ is the KL-divergence. The derivation of Eq.~\ref{equation:elbo-obj} is provided in Appendix~\ref{sec:ELBO-derive}. This formulation can be interpreted as an Evidence Lower Bound (ELBO). The first term encourages accurate target prediction under graph-feasible intents, while the KL term regularizes the model's inferred intent distribution to align with the graph-conditioned prior.

\vspace{5pt}
\noindent\textbf{Connection to Context Engineering and Latent Reasoning.}
In practice, the first term in Eq.~\ref{equation:elbo-obj} encourages the model to predict the target item $i^*$ conditioned on both the user history $H$ and the intent prototype $c$. This requires injecting future candidate knowledge into the model's input or conditioning pathway, which closely parallels \emph{Context Engineering} for LLMs~\cite{abs-2507-13334}, where external or structured context is provided to guide prediction and reduce uncertainty. In our setting, the candidate set derived from the interaction graph serves as structured contextual knowledge that narrows the model's predictive focus to collaboratively plausible regions.

Meanwhile, the intent distribution $p_\theta(c \mid H)$ is implicitly induced by latent reasoning states through model's output layer, by projecting each latent state onto item probability simplex. Minimizing the KL divergence term aligns this induced distribution with the graph-conditioned prior, constraining each reasoning step to remain within the collaborative manifold supported by $\mathcal{C}(k)$. The variational regularization explicitly limits the freedom of latent refinement and mitigates latent drift during iterative reasoning.

\subsection{Local Graph Smoothness by KL Distillation}
\label{sec:graphsmoothness}

The variational objective in Eq.~\ref{equation:elbo-obj} introduces a KL-divergence regularization term. To justify why this specific term is effective against latent drift, we provide a theoretical analysis in this subsection. We prove that minimizing this KL term induces a gradient flow that promotes ``local graph smoothness''. This analysis bridges the gap between our probabilistic objective and the geometric manifold constraint, showing how the variational objective explicitly confines the reasoning trajectory by collaboratively feasible items.

\begin{proposition}[Local Graph Smoothness Induced by KL Distillation]
\label{prop:local_smoothness}
Let $\mathcal{C}$ be a finite candidate set and $\mathbf{e}_c \in \mathbb{R}^d$ denote the embedding of item $c \in \mathcal{C}$.
Given a reasoning state $\mathbf{r} \in \mathbb{R}^d$, define the induced predictive distribution as:
\begin{equation*}
P(c \mid H)
= \frac{\exp\!\left(\mathbf{r}^\top \mathbf{e}_c\right)}
{\sum_{c' \in \mathcal{C}} \exp\!\left(\mathbf{r}^\top \mathbf{e}_{c'}\right)}, \quad c \in \mathcal{C}. 
\end{equation*}
Let $Q$ be any fixed teacher distribution supported on $\mathcal{C}$.
Then the KL distillation loss
$\mathcal{L}(\mathbf{r}) = D_\mathrm{KL}\!\left(Q \,\|\, P(\cdot \mid H)\right)$
is differentiable with respect to $\mathbf{r}$, with gradient
\begin{equation*}
\nabla_{\mathbf{r}} \mathcal{L}(\mathbf{r})
=
\mathbb{E}_{P(\cdot \mid H)}[\mathbf{e}_c]
-
\mathbb{E}_{Q}[\mathbf{e}_c]. 
\end{equation*}
\end{proposition}
The proof of Proposition~\ref{prop:local_smoothness} is provided in Appendix~\ref{sec:proof_local_smoothness}.

\vspace{5pt}
\noindent\textbf{Interpretation.}
The distribution $P(\cdot \mid H)$ defines a point on the probability simplex over $\mathcal{C}$, and its embedding
expectation $\mathbb{E}_{P}[\mathbf{e}_c]$ lies in the convex hull of candidate embeddings 
$\mathrm{conv}\{\mathbf{e}_c : c \in \mathcal{C}\}$.
The teacher expectation $\mathbb{E}_{Q}[\mathbf{e}_c]$ lies in the same region. 
Proposition~\ref{prop:local_smoothness} shows that KL distillation induces a gradient flow that directly moves the prediction barycenter toward the teacher barycenter within this graph-restricted convex hull.
As a result, latent reasoning states are encouraged to evolve such that their induced predictions remain confined to a graph-local embedding region defined by $\mathcal{C}$. In ManCAR, $\mathcal{C}$ is a graph-conditioned candidate set derived from the most recent interactions, and $Q$ is a scheduled teacher distribution supported on this set. Together, they impose a \emph{local graph smoothness prior} on latent reasoning: each refinement step reduces uncertainty while remaining restricted to collaboratively reachable items. Progressive sharpening of the teacher distribution yields a stable coarse-to-fine trajectory on the simplex, mitigating latent drift during multi-step reasoning.

\subsection{Implementation of ManCAR Objective}
\label{sec:implementation}

Having established the theoretical validity of the variational training objective, we now turn to the implementation of ManCAR. 
In this subsection, we describe how the variational objective in Eq.~\ref{equation:elbo-obj} is instantiated in practice. 
Our implementation is designed to faithfully realize the graph-conditioned manifold constraint while remaining compatible with standard latent reasoning paradigms.

\vspace{5pt}
\noindent\textbf{Item Interaction Graph.}
We construct the item interaction graph using a standard Swing-style item-to-item co-interaction algorithm that is widely adopted in industrial systems~\cite{abs-2010-05525}. Each node corresponds to an item, and weighted edges encode collaborative strength measured by co-interaction frequency. As this graph construction follows established practice and is not a contribution of this work, we defer algorithmic details to Appendix~\ref{sec:swing_build}. Concretely, the resulting graph $\mathcal{G}$ associates each directed edge $(i\rightarrow j)$ with a weight $w_{ij}$, indicating the strength of collaborative relevance of item $j$ with respect to item $i$.

\vspace{5pt}
\noindent\textbf{Teacher Prior Construction.}
The variational objective in Equation~\eqref{equation:elbo-obj} requires a teacher distribution $q(c \mid I_n, \mathcal{G})$ defined over the candidate set $\mathcal{C}(k)$. This distribution encodes prior knowledge about plausible intent prototypes that are collaboratively reachable from the user's most recent interactions. In practice, we consider the following strategy to construct the teacher prior efficiently:

\vspace{5pt}
\noindent\emph{Rank-Based Distribution Mass Assignment (RDMA).}
We construct the teacher prior based on the relative ranking of candidates. The target item $i^*$ is always assigned a rank $0$, while the remaining candidates in $\mathcal{C}(k)$ are ranked in descending order of their graph edge weights $w_{I_{n},c}$ (from a recent interacted item in $I_n$ to its neighbor $c$), receiving ranks $[1,2,\dots]$. The probability mass is then assigned using a softmax over negative ranks,
\begin{equation}
\label{eqn:q_strategy_2}
q(c \mid I_{n}, \mathcal{G})
=
\frac{\exp(-\mathrm{rank}(c)/\gamma)}
{\sum_{n \in \mathcal{C}(k)} \exp(-\mathrm{rank}(n)/\gamma)},
\end{equation}
where $\gamma>0$ controls the sharpness of the teacher distribution. Using negative ranks ensures that higher-ranked (i.e., more strongly connected) items and the target receive the larger probability mass.

This strategy ensures that the teacher prior is strictly supported on the collaborative neighborhood while emphasizing the target item. By default, we assign zero probability mass to non-candidate items, resulting in a teacher distribution that lies on a sparse region of the item probability simplex with only a small number of active entries. Alternatively, a small amount of probability mass can be distributed over non-candidate items as a form of label smoothing, which we find does not materially affect the main conclusions.
\textbf{A dynamic scheduling} is applied to the teacher prior distributions to guide latent reasoning progressively toward the target by adjusting the concentration of the teacher distribution across reasoning steps. As this scheduling mechanism is closely tied to adaptive test-time termination, we defer its detailed formulation to Sec.~\ref{sec:adaptive_reasoning}.

\vspace{5pt}
\noindent\textbf{Training Loss.}
The complete training objective instantiates the ELBO derived in Eq.~\ref{equation:elbo-obj} and consists of a target prediction loss and a graph-conditioned manifold regularization loss, both applied at each reasoning step.

\vspace{5pt}
\noindent\emph{Latent Reasoning and Decoding.}
At reasoning step $t'$, the model produces a latent reasoning state
$\mathbf r_{t'} \in \mathbb R^d$, obtained by iteratively refining the
initial state $\mathbf r_1 = \mathbf h_{T-1}=f_{\theta}(H)$ through a shared reasoning module $f_{\theta}(\cdot)$. Concretely, the refinement follows $r_{t'}=f_{\theta}(H;r_{1:t'-1})$. Then each latent reasoning state $\mathbf{r}_{t'}$ is projected onto the item space to produce logits:
\begin{equation*}
\mathbf{z}_{t'} = \mathbf{r}_{t'}^\top \mathbf{E},
\end{equation*}
where $\mathbf{E}\in\mathbb{R}^{d\times|\mathcal{I}|}$ denotes the item embedding matrix. The logits define an item probability distribution over $\mathcal{I}$ via a temperature-scaled softmax,
\begin{equation}
p_\theta^{(t')}(i \mid H)
=
\frac{\exp\!\left(\mathbf{z}_{t',i} / \tau_{t'}\right)}
{\sum_{j \in \mathcal{I}} \exp\!\left(\mathbf{z}_{t',j} / \tau_{t'}\right)}. \nonumber
\end{equation}
We use this same distribution to represent the intent distribution $p_\theta^{(t')}(c \mid H)$ by restricting $c$ to the candidate set $\mathcal{C}(k)$.

\vspace{5pt}
\noindent\emph{Target Prediction Loss.}
Following the ELBO in Eq.~\ref{equation:elbo-obj}, the target prediction term at reasoning step $t'$ is given by

\begin{equation*}
\mathcal L_{\mathrm{main}}^{(t')}
=
-
\mathbb E_{c \sim q(\cdot \mid I_{n},\mathcal G)}
\big[
\log p_\theta^{(t')}(i^* \mid H, c)
\big].
\end{equation*}
In practice, explicitly marginalizing over $c \in \mathcal C(k)$ at each reasoning step is costly and unnecessary. Instead, we adopt a standard conditioning strategy by exposing the entire candidate set $\mathcal C(k)$ to the model as additional contexts (like context engineering). Concretely, we approximate the above expectation by:
\begin{equation}
\label{eqn:loss_main_t}
\mathcal L_{\mathrm{main}}^{(t')}
=
-
\log p_\theta^{(t')}(i^* \mid H, \mathcal C(k)),
\end{equation}
where $\mathcal C(k)$ is injected as auxiliary input alongside the history $H$.

\vspace{5pt}
\noindent\emph{Graph-Conditioned Manifold Regularization.}
To enforce graph-conditioned feasibility, the induced distribution at each step is regularized toward the teacher prior via

\begin{equation}
\label{eqn:loss_kl_t}
\mathcal{L}_{\mathrm{reg}}^{(t')}
=
D_{\mathrm{KL}}\!\left(
q(c \mid I_{n}, \mathcal{G})
\;\|\;
p_\theta^{(t')}(c \mid H)
\right).
\end{equation}
This term restricts latent reasoning trajectories to remain within the collaborative manifold defined by $\mathcal{C}(k)$ and mitigates latent drift. To reduce computation, the student distribution $p_\theta^{(t')}(c \mid H)$ can be obtained from the same forward pass used for target prediction 

\vspace{5pt}
\noindent\emph{Overall Objective.}
The complete training objective is as follows: % aggregates step-wise losses
\begin{equation}
\label{eqn:loss_overall}
\mathcal{L}
=
\sum_{t'=1}^{T'}
\left(
\mathcal{L}_{\mathrm{main}}^{(t')}
+
\lambda \mathcal{L}_{\mathrm{reg}}^{(t')}
\right),
\end{equation}
where $\lambda$ controls the strength of graph-conditioned regularization. Appendix~\ref{sec:app_mancar_algo} provides the detailed training algorithm of ManCAR.

\subsection{Training Scheduling and Adaptive Test-Time Reasoning}
\label{sec:adaptive_reasoning}

We now introduce a dynamic scheduling mechanism that sharpens the teacher distribution over time and present a theoretical proposition regarding bounded error. This analysis not only guides the training schedule but also naturally motivates our convergence-based stopping criterion, enabling adaptive test-time reasoning. Appendix~\ref{sec:app_mancar_algo} provides the detailed adaptive reasoning algorithm.

\begin{proposition}[Continuation tracking under contraction and bounded teacher drift]
\label{prop:continuation_tracking}
Fix a query $H$ and a finite candidate set $\mathcal{C}=\mathcal{C}(k)$. Let $\Delta(\mathcal{C})$ denote the item probability simplex over $\mathcal{C}$. Consider distribution sequences $\{p_{t'}\}_{t'\ge 1}\subset \Delta(\mathcal{C})$ (student) and
$\{q_{t'}\}_{t'\ge 1}\subset \Delta(\mathcal{C})$ (scheduled teacher), and let
\begin{equation}
d_{\mathrm{TV}}(p,q) := \frac{1}{2}\sum_{c\in\mathcal{C}} |p(c)-q(c)|  \nonumber
\end{equation}
be the total variation distance. We assume:
\begin{enumerate}[leftmargin=*]
\item \textbf{Stepwise contraction toward the current teacher:}
There exists $\lambda\in(0,1)$ such that for all $t'\ge 1$,
\begin{equation}
d_{\mathrm{TV}}(p_{t'+1}, q_{t'}) \le (1-\lambda)\, d_{\mathrm{TV}}(p_{t'}, q_{t'}).  
\label{eq:assump_contraction}
\end{equation}

\item \textbf{Bounded teacher drift (controlled schedule):}
There exists $\delta\ge 0$ such that for all $t'\ge 1$,
\begin{equation}
d_{\mathrm{TV}}(q_{t'+1}, q_{t'}) \le \delta.  
\label{eq:assump_drift}
\end{equation}
\end{enumerate}

Then for all $t'\ge 1$,
\begin{equation}
d_{\mathrm{TV}}(p_{t'}, q_{t'})
\le
(1-\lambda)^{t'-1}\, d_{\mathrm{TV}}(p_1,q_1) + \frac{\delta}{\lambda}.  
\label{eq:tracking_bound}
\end{equation}
\end{proposition}

We provide the proof of Proposition~\ref{prop:continuation_tracking} in Appendix~\ref{sec:proof_continuation_tracking}. Proposition~\ref{prop:continuation_tracking} shows that, if each refinement step contracts the student toward the current teacher and the teacher distribution evolves smoothly, then the student distribution can track a progressively changing teacher with bounded error. This provides a formal motivation for using a coarse-to-fine teacher schedule during training.

The proposition is stated in total variation (TV) distance to leverage its metric properties. In ManCAR, refinement is trained via KL distillation rather than TV. To bridge this gap, we invoke
Pinsker's inequality, which guarantees that for any distributions $p,q$ on $\mathcal{C}$,
\begin{equation}
d_{\mathrm{TV}}(p,q)
\le
\sqrt{\tfrac{1}{2} D_{\mathrm{KL}}(q \,\|\, p)}. \nonumber
\end{equation}
Hence, minimizing the KL loss ensures a small student-teacher mismatch in TV, providing a conservative stability guarantee for the continuation tracking behavior described in
Proposition~\ref{prop:continuation_tracking}.

\vspace{5pt}
\noindent\textbf{Teacher Scheduling Strategies.}
We extend the teacher construction (Sec.~\ref{sec:implementation}) into a scheduled
mechanism that generates a smoothly evolving sequence of teacher distributions across reasoning steps.

\vspace{5pt}
\noindent\emph{Adjustment for Strategy RDMA.} We define the teacher as:
\begin{equation}
q_{t'}(c)\propto \exp\!\big(-\mathrm{rank}(c)/\gamma_{t'}\big),
\qquad
\gamma_{t'}=\gamma_{\text{base}}\cdot(T'-t'+1), \nonumber
\end{equation}
with $\gamma_{\text{base}}\ge 1$. As $\gamma_{t'}$ decreases linearly, the teacher distribution
transitions smoothly from a diffuse graph-aware prior to a sharply peaked distribution centered
on the target. Properly tuning $\gamma_{\text{base}}$ and total steps $T$ yields a smoothly evolving teacher distribution with bounded drift across refinement steps (satisfying Assumption (2) in Proposition~\ref{prop:continuation_tracking}).

\vspace{5pt}
\noindent\textbf{Connection to Adaptive Test-Time Reasoning.}
The continuation view provided by Proposition~\ref{prop:continuation_tracking} directly motivates
adaptive termination at test time. Since the student distribution $p_{t'}$ tracks the scheduled
teacher with bounded error, convergence of successive student distributions indicates that further
refinement yields diminishing returns. We therefore terminate reasoning early when the change
between consecutive steps falls below a threshold, e.g.,
\begin{equation}
D_{\mathrm{KL}}\!\left(p_{t'-1} \,\|\, p_{t'}\right) < \varepsilon. \nonumber
\end{equation}

\vspace{5pt}
\noindent\textbf{Scheduling the Main Prediction Loss.}
In addition to scheduling the teacher prior, we apply a step-dependent temperature schedule to the main target prediction loss to control the magnitude of distributional updates induced by target supervision at each reasoning step. Concretely, we use an exponential temperature schedule: 

\begin{equation*}
\tau_{t'} = \tau_{\mathrm{base}} \cdot t'^{\alpha},
\end{equation*}
where $\alpha>1$ controls $\tau_{t'}$'s increasing magnitude, which yields an increasing temperature sequence across reasoning steps with a flexible initial temperature base $\tau_{\mathrm{base}}$.

From the continuation perspective formalized in Proposition~\ref{prop:continuation_tracking},
supervising all refinement steps with an identical, sharply peaked target loss may induce overly
large early updates that violate the bounded-drift and contraction conditions, potentially
destabilizing manifold-constrained reasoning.
By starting with a low effective temperature, early refinement steps are encouraged to make
conservative progress (near the local neighborhood of recent user interactions), while remaining within the graph-consistent manifold.

This design is related in spirit to progressive refinement losses such as the PRL mechanism in ReaRec~\cite{abs-2503-22675}, but differs in directionality. Whereas ReaRec adopts a \textbf{decreasing} temperature schedule to accelerate early-stage convergence, our \textbf{increasing}-temperature design aligns with our theoretical analysis and supports stable multi-step reasoning and adaptive test-time termination.

\vspace{5pt}
\noindent\textbf{Stabilizing Optimization via Latent State Norm Rescaling.}
In addition to scheduling-based control, we apply normalization to stabilize multi-step latent reasoning.
After each refinement step, we rescale the latent reasoning state as:
\[
\mathbf{h} \;\leftarrow\; \phi \cdot \frac{\mathbf{h}}{\|\mathbf{h}\|}\cdot \mathrm{avg}(\mathbf{E}),
\]
where $\mathrm{avg}(\mathbf{E})$ denotes the average norm of item embeddings, and $\phi$ is a learnable affine scaling parameter.

This operation aligns the scale of latent states with that of the item embedding space. This rescaling alleviates the burden on the Transformer to simultaneously accommodate heterogeneous modalities with mismatched norms between original input items and latent reasoning states.
By keeping latent states on a scale comparable to item embeddings, this normalization mitigates empirical norm growth with (recursive) depth~\cite{layernormscaling}, improves stability in long-horizon reasoning, and complements the manifold-constrained and continuation-based design of ManCAR. In particular, it helps maintain a well-conditioned softmax geometry during refinement, which empirically supports the stepwise contraction behavior assumed in Proposition~\ref{prop:continuation_tracking} (facilitating Assumption (1)).

\section{Experiments}

Our empirical study is guided by the following \textbf{research questions}: \emph{(1) Overall performance.}
How does ManCAR perform compared with strong sequential recommendation baselines across standard
benchmarks? \emph{(2) Effect of teacher scheduling and adaptive reasoning.} How do teacher scheduling and adaptive termination shape step-wise refinement behavior and enable near-ceiling performance during inference? \emph{(3) Ablation analysis.}
What is the impact of individual components in ManCAR? \emph{(4) Parameter sensitivity.} How sensitive is ManCAR to key hyperparameters? \emph{(5) KL-Based Halting Analysis.} How does the KL divergence between steps reflect the stability of the reasoning trajectory?
\emph{(6) Attention Visualization Analysis.} What do attention patterns reveal about the information flow within the manifold-constrained design?

\subsection{Experimental Setup}

\subsubsection{Datasets and Preprocess}
\label{sec:appx_data}

\begin{table}[t]
\centering
\small
\caption{Dataset statistics.}
\vspace{-10pt}
\label{tab:dataset_stats}
\resizebox{\linewidth}{!}{
\begin{tabular}{lccccc}
\toprule
Dataset & \#Users & \#Items & \#Interactions & \begin{tabular}{@{}c@{}}\#Avg.\\ Inter./User\end{tabular} & \begin{tabular}{@{}c@{}}\#Avg.\\ Inter./Item\end{tabular}  \\
\midrule
CDs     & 35,238  & 87,969  & 943,399   & 26.77 & 10.72 \\
Video   & 54,001  & 22,735  & 562,193   & 10.41 & 24.73 \\
Office  & 118,617 & 64,679  & 1,116,568   & 9.41  & 17.3   \\
Arts    & 112,401 & 77,596  & 1,180,363   & 10.50  & 15.21   \\
Music   & 15,685  & 21,171  &  221,980   & 14.15  & 10.49   \\
Toys    & 244,231 & 134,676 & 2,515,203 & 10.30 & 18.68  \\
Grocery & 246,304 & 119,860 & 2,788,430 & 11.32  & 23.26  \\
\bottomrule
\end{tabular}
}
\end{table}

We evaluate ManCAR on seven sub-category datasets from the Amazon 2023 Reviews corpus~\cite{abs-2403-03952}: CDs \& Vinyl (CDs), Video \& Games (Video), Office Products (Office), Arts, Crafts \& Sewing (Arts),
Grocery \& Gourmet Food (Grocery), Musical Instruments (Music), and Toys \& Games (Toys). 
Tab.~\ref{tab:dataset_stats} provides data statistics. 
Following prior work~\cite{abs-2503-22675,abs-2505-16865,abs-2601-03153}, user--item interactions with
ratings above 3 are treated as positive feedback.

To improve data quality, we remove users with fewer than 10 interactions in CDs and fewer than 5 interactions in the remaining datasets. We adopt the official absolute-timestamp split provided by the corpus.\footnote{https://amazon-reviews-2023.github.io/data\_processing/5core.html} Consistent with previous studies~\cite{abs-2503-22675,abs-2505-16865,abs-2601-03153}, we truncate each user’s interaction history to a maximum length of 50.

\subsubsection{Evaluation Metrics.}
 
To evaluate the performance of our proposed model and the baselines, we employ two widely-used metrics: Recall@$K$ and Normalized Discounted Cumulative Gain (NDCG@$K$), with $K \in \{5, 10\}$.
Specifically, Recall@$K$ measures the model's ability to include the ground-truth item within the top-$K$ recommendation list, reflecting its retrieval coverage.
NDCG@$K$ further assesses the ranking quality by assigning higher weights to items at higher positions, thereby rewarding models that prioritize the correct item in more prominent ranks.

\subsubsection{Baselines.}
We compare ManCAR with representative state-of-the-art baselines spanning different modeling paradigms. Specifically, we include:

\emph{(1) SASRec}~\cite{KangM18}: Utilizing a unidirectional Transformer encoder, SASRec represents users by the final item in their interaction sequence.

\emph{(2) BERT4Rec}~\cite{sun2019bert4rec}: adopts a bidirectional approach inspired by BERT, training the model to reconstruct masked items within the sequence.

\emph{(3) ContextBERT4Rec}: extends BERT4Rec by using the same context engineering as ManCAR.

\emph{(4) ReaRec-ERL}~\cite{abs-2503-22675}: As a pioneer in latent space reasoning, this model treats the reasoning process as a collective trajectory. Instead of relying on a single state, it synthesizes the implicit information from all autoregressive steps using a mean pooling mechanism to form a comprehensive user representation.

\emph{(5) ReaRec-PRL}~\cite{abs-2503-22675}: In contrast to ERL, this variant emphasizes iterative optimization. It leverages contrastive learning with noise injection to progressively distill the latent representation, discarding intermediate states to rely solely on the converged output of the final reasoning step.

\emph{(6) LARES}~\cite{abs-2505-16865}: This framework introduces pre-blocks and core-blocks. To maximize reasoning fidelity, it adopts a hybrid training pipeline that sequentially applies self-supervised pre-training followed by reinforcement learning-based fine-tuning.

\emph{(7) PLR}~\cite{abs-2601-03153}: A width-scaled (parallel) latent reasoning framework for sequential recommendation that launches multiple parallel reasoning streams via learnable trigger tokens, enforces inter-stream diversity with global reasoning regularization, and adaptively fuses the stream outputs (mixture-of-streams) to improve next-item prediction.

\subsubsection{Implementation.}
We conduct all experiments on eight NVIDIA 3090 GPUs. To ensure a fair comparison, we set the embedding size and batch size for all methods to 256 and 512, respectively. We optimize all models using the Adam optimizer with a learning rate of 0.001. For baselines without open-source code, we conducted our own implementation. For those with available source code, we utilized the official implementations. All baselines were tuned via grid search based on the hyperparameters specified in their original papers, and the optimal results are reported. To mitigate overfitting, we employ early stopping, terminating training if NDCG@10 on the validation set shows no improvement for 5 consecutive epochs.
Following prior work \cite{abs-2503-22675}, we adopt a pre-norm Transformer backbone. It consists of two Transformer layers, each with two-head multi-head self-attention and GeLU activation.

\begin{table*}[t]
\caption{Performance comparison on seven datasets. The best results are in \textbf{bold} and the second best results are \underline{underlined}.}
\vspace{-10pt}
\label{tab:performance_cross}
\centering
\small
\begin{tabular*}{\textwidth}{@{\extracolsep{\fill}}ll|ccccccc|cc}

\toprule
Dataset & Metric & SASRec & BERT4Rec & ContextBERT4Rec & ERL & PRL & PLR & LARES & \textbf{ManCAR} & Uplift \\
\midrule
\multirow{4}{*}{CDs} 
& NDCG@5  & 0.0098 & 0.0110 & 0.0148 & 0.0099 & 0.0122 & 0.0131 & \underline{0.0159} & \textbf{0.0198} & 24.53~\% \\
& NDCG@10 & 0.0132 & 0.0130 & 0.0182 & 0.0129 & 0.0149 & 0.0168 & \underline{0.0192} & \textbf{0.0282} & 46.88~\% \\
& Recall@5    & 0.0181 & 0.0207 & 0.0243 & 0.0190 & 0.0227 & \underline{0.0249} & 0.0235 & \textbf{0.0346} & 38.96~\% \\
& Recall@10   & 0.0286 & 0.0269 & \underline{0.0391} & 0.0283 & 0.0315 & 0.0363 & 0.0351 & \textbf{0.0516} & 31.97~\% \\
\midrule
\multirow{4}{*}{Video} 
& NDCG@5  & 0.0083 & 0.0070 & 0.0119 & 0.0149 & \underline{0.0165} & 0.0159 & 0.0162 & \textbf{0.0217} & 31.52~\% \\
& NDCG@10 & 0.0104 & 0.0091 & 0.0163 & 0.0187 & 0.0215 & 0.0210 & \underline{0.0219} & \textbf{0.0275} & 25.57~\% \\
& Recall@5    & 0.0164 & 0.0131 & 0.0190 & 0.0217 & 0.0253 & \underline{0.0312} & 0.0277 & \textbf{0.0339} & 8.65~\% \\
& Recall@10   & 0.0226 & 0.0196 & 0.0330 & 0.0336 & 0.0410 & \underline{0.0467} & 0.0455 & \textbf{0.0521} & 11.56~\% \\
\midrule
\multirow{4}{*}{Office} 
& NDCG@5  & 0.0063 & 0.0061 & 0.0076 & 0.0079 & 0.0071 & 0.0082 & \underline{0.0090} & \textbf{0.0108} & 20.00~\% \\
& NDCG@10 & 0.0077 & 0.0078 & 0.0093 & 0.0107 & 0.0096 & 0.0109 & \underline{0.0115} & \textbf{0.0133} & 15.65~\% \\
& Recall@5    & 0.0113 & 0.0109 & 0.0129 & 0.0144 & 0.0127 & 0.0139 & \underline{0.0161} & \textbf{0.0174} & 8.07~\% \\
& Recall@10   & 0.0157 & 0.0160 & 0.0184 & 0.0229 & 0.0205 & 0.0224 & \underline{0.0240} & \textbf{0.0250} & 4.17~\% \\
\midrule
\multirow{4}{*}{Arts} 
& NDCG@5  & 0.0019 & 0.0014 & 0.0044 & 0.0052 & 0.0072 & 0.0068 & \underline{0.0077} & \textbf{0.0087} & 12.99~\% \\
& NDCG@10 & 0.0026 & 0.0023 & 0.0059 & 0.0082 & 0.0097 & 0.0092 & \underline{0.0101} & \textbf{0.0114} & 12.87~\% \\
& Recall@5    & 0.0037 & 0.0028 & 0.0082 & 0.0086 & 0.0118 & 0.0106 & \underline{0.0119} & \textbf{0.0141} & 18.49~\% \\
& Recall@10   & 0.0057 & 0.0055 & 0.0129 & 0.0177 & \underline{0.0196} & 0.0181 & \underline{0.0196} & \textbf{0.0225} & 14.80~\% \\
\midrule
\multirow{4}{*}{Music} 
& NDCG@5  & 0.0017 & 0.0014 & 0.0049 & 0.0066 & \underline{0.009} & 0.0058 & {0.0063} & \textbf{0.0097} & 7.78~\% \\
& NDCG@10 & 0.0027 & 0.002 & 0.0059 & 0.0095 & \underline{0.0117} & 0.0084 & {0.0092} & \textbf{0.0120} & 2.56~\% \\
& Recall@5    & 0.0033 & 0.0026 & 0.0086 & 0.0124 & \underline{0.0138} & 0.0111 & {0.0117} & \textbf{0.0147} & 6.52~\% \\
& Recall@10   & 0.0061 & 0.0046 & 0.0118 & 0.0215 & \textbf{0.0225} & 0.0191 & {0.0205} & \underline{0.0217} & -3.56~\% \\

\midrule
\multirow{4}{*}{Toys} 
& NDCG@5  & 0.0037 & 0.0045 & 0.0051 & 0.0054 & 0.0046 & 0.0052 & \underline{0.0061} & \textbf{0.0086} & 40.98~\% \\
& NDCG@10 & 0.0048 & 0.0059 & 0.0063 & 0.0073 & 0.0064 & 0.0068 & \underline{0.0082} & \textbf{0.0108} & 31.71~\% \\
& Recall@5    & 0.0072 & 0.0086 & 0.0089 & 0.0102 & 0.0082 & 0.0096 & \underline{0.0115} & \textbf{0.0136} & 18.26~\% \\
& Recall@10   & 0.0104 & 0.0127 & 0.0126 & 0.0160 & 0.0136 & 0.0148 & \underline{0.0180} & \textbf{0.0203} & 12.78~\% \\

\midrule
\multirow{4}{*}{Grocery} 
& NDCG@5    & 0.0028 & 0.0026 & 0.0044 & 0.0062 & 0.0072 & \underline{0.0078} & 0.0063 & \textbf{0.0095} & 21.79~\% \\
& NDCG@10   & 0.0038 & 0.0035 & 0.0059 & 0.0087 & 0.0098 & \underline{0.0105} & 0.0088 & \textbf{0.0118} & 12.38~\% \\
& Recall@5  & 0.0055 & 0.0085 & 0.0079 & 0.0102 & 0.0118 & \underline{0.0124} & 0.0119 & \textbf{0.0149} & 20.16~\% \\
& Recall@10 & 0.0085 & 0.0078 & 0.0126 & 0.0180 & 0.0197 & \underline{0.0218} & 0.0197 & \textbf{0.0220} & 0.92~\% \\

\bottomrule
\end{tabular*}
\end{table*}

\subsection{Overall Performance (Tab.~\ref{tab:performance_cross}).}

\noindent
\textbf{ManCAR outperforms all baselines.}
ManCAR achieves the best performance across all datasets and evaluation metrics, demonstrating consistent improvements in both ranking and retrieval quality. Compared with the second-best method on each dataset, ManCAR delivers up to a 46.88\% relative
improvement on certain metrics. Notably, the gains are more pronounced on NDCG, indicating that ManCAR is particularly effective at ranking relevant items higher, which reflects a stronger ability to capture and refine user intent.

\vspace{5pt}
\noindent
\textbf{ContextBERT4Rec outperforms BERT4Rec}, highlighting the benefit of graph-induced context.
ContextBERT4Rec augments the input sequence with the same graph-conditioned candidate set used by ManCAR, enabling the model to leverage collaborative signals beyond the independent raw user interaction sequence. Its consistent improvement over BERT4Rec suggests that incorporating graph-induced
context serves as an effective form of context engineering for sequential recommendation.

\vspace{5pt}
\noindent
\textbf{Explicit latent reasoning consistently improves sequential recommendation.}
ContextBERT4Rec represents the strongest non-reasoning baseline by incorporating graph-conditioned context into the input.
Across all datasets, ManCAR achieves notable gains over ContextBERT4Rec, demonstrating that explicit multi-step reasoning provides additional modeling capacity beyond contextual encoding alone. More broadly, all reasoning-based methods (ManCAR, ERL, PRL, PLR, and LARES) outperform non-reasoning baselines like SASRec and BERT4Rec, suggesting that iterative refinement of intermediate
hypotheses enables more effective uncertainty resolution and user intent modeling, particularly in sparse or challenging settings.

\begin{table}[t]
    \centering
    \caption{Best performing step setting of reasoning-based methods on four datasets. See Appendix~\ref{sec:allocation} for full results.}
    \label{tab:step_partial}
    \vspace{-10pt}
    \resizebox{\linewidth}{!}{
    \begin{tabular}{ll|cccc|c}
    \hline
    Dataset                  & Reason step & ERL & PRL & PLR & LARES & ManCAR \\ \hline
    \multirow{2}{*}{CDs}     & Train step  & 2   & 2   & 3 &4   & 5      \\
                             & Infer step  & 2   & 2   & 3 &4  & 1.84   \\ \hline
    \multirow{2}{*}{Arts}    & Train step  & 2   & 2   & 1 &4  & 1      \\
                             & Infer step  & 2   & 2   & 1 &4  & 1      \\ \hline
    \multirow{2}{*}{Toys}    & Train step  & 1   & 2   & 1 &4  & 4      \\
                             & Infer step  & 1   & 2   & 1 &4  & 3.58   \\ \hline
    \multirow{2}{*}{Grocery} & Train step  & 2   & 2   & 3 &4  & 2      \\
                             & Infer step  & 2   & 2   & 3 &4  & 1.74   \\ \hline
    \end{tabular}
    }
    \vspace{-10pt}
\end{table}

\vspace{5pt}
\noindent
\textbf{ManCAR consistently outperforms existing latent reasoning approaches.}
Across all datasets, ManCAR achieves consistent gains over prior reasoning-based methods such as ERL, PRL, PLR, and LARES.
While these methods introduce latent refinement or progressive reasoning, they typically lack explicit constraints on how reasoning trajectories evolve. In contrast, ManCAR integrates graph-conditioned manifolds, scheduled teacher supervision, and adaptive test-time control, which together provide a more structured and stable reasoning process. By explicitly controlling the feasible manifold region and stepwise dynamics of latent refinement, ManCAR
is better able to exploit collaborative signals and avoid unstable or suboptimal reasoning paths, leading to better performance with varying data sparsity and sequence lengths.

\vspace{5pt}
\noindent
\textbf{Performance gains increase with higher interaction density.}
ManCAR exhibits larger performance margins over the second-best baseline on datasets with higher interaction density (average interactions per item).
For instance, improvements are more pronounced on Video and Toys than on Music and Arts. This trend suggests that ManCAR benefits from reduced sparsity, where multi-step reasoning can more effectively refine user intent by leveraging a more reliable item interaction graph and richer collaborative signals. When interactions are sparse, graph edge connection become noisier, which limits the advantage of graph-conditioned reasoning over strong baselines. Improving robustness under limited preference evidences (cold start) is left for future work.

\begin{figure}[t]  
  \centering
  \includegraphics[width=\columnwidth]{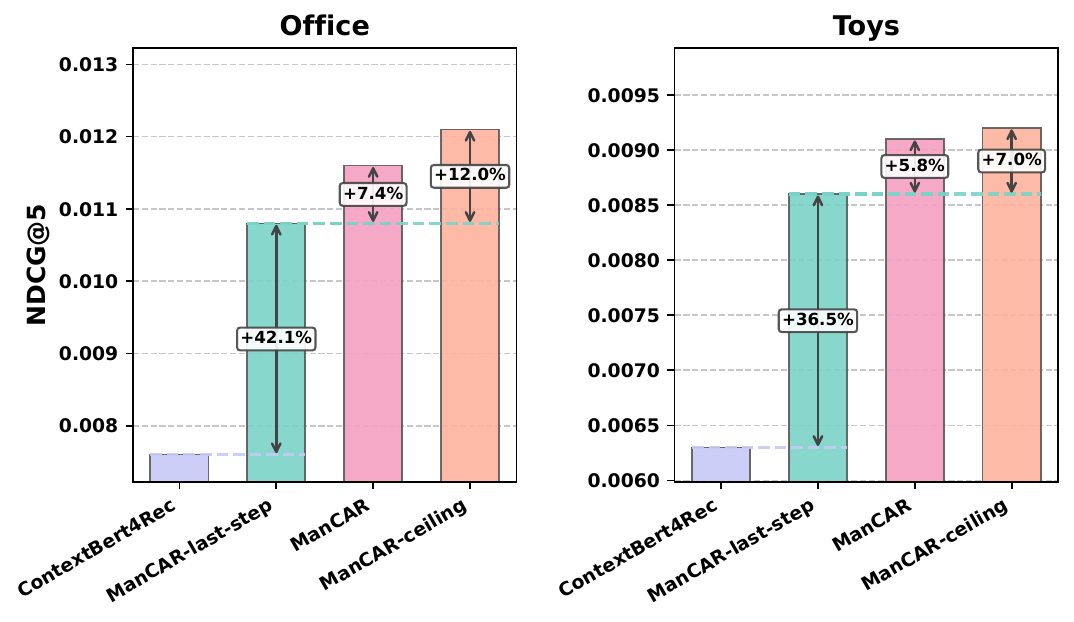}
  \vspace{-20pt}
  \caption{Performance ceiling analysis on Office and Toys.}
  \vspace{-10pt}
  \label{fig:ndcg_ceiling}
\end{figure}

\subsection{In-Depth Analysis in Adaptive Reasoning}
\label{sec:in-depthana}

We analyze ManCAR's adaptive reasoning ability from two angles. 

\vspace{5pt}
\noindent\textbf{Data-Aware Train-Test Compute Allocation.}
Tab.~\ref{tab:step_partial} summarizes the reasoning-step configurations at which different reasoning methods achieve their best performance. Baselines adopt \emph{identical and shallow} reasoning depths at both stages,
typically limited to 2-3 steps, regardless of data characteristics. This indicates that reasoning depth is treated as a static architectural hyperparameter.

In contrast, ManCAR exhibits data-aware and asymmetric train-test computation. The optimal number of training and inference steps varies substantially across datasets, reflecting differences in data sparsity and sequence complexity.
On datasets with complex interaction patterns, such as CDs and Toys, ManCAR employs deeper reasoning and achieves significantly larger performance gains, while prior methods are unable to
adapt beyond 3 steps. Conversely, on simpler datasets such as Arts and Grocery, ManCAR stops early at inference,
avoiding unnecessary computation while still outperforming baselines which over-allocate reasoning steps.

Overall, these results indicate that ManCAR performs genuine iterative refinement with adaptive inference depth, enabling an effective balance between reasoning expressiveness and computational efficiency across diverse data properties.

\vspace{5pt}
\noindent\textbf{Near-Optimal Reasoning through the Lens of Ceiling Performance Analysis.} Fig.~\ref{fig:ndcg_ceiling} showcases a step-wise performance analysis of ManCAR.
We report three variants: (i) the prediction from the final reasoning step (ManCAR-last-step), (ii) adaptive halting based on convergence (ManCAR), and (iii) an oracle ceiling that selects the best-performing step per sample using ground-truth labels (ManCAR-ceiling). These results are compared with ContextBERT4Rec, a non-reasoning variant of ManCAR.

When ManCAR is forced to use a fixed (symmetric) number of reasoning steps, performance degrades relative to adaptive halting, though it remains closer to the ceiling than the non-reasoning variant. In contrast, adaptive reasoning consistently outperforms the symmetric setting and achieves performance that is very close to the oracle ceiling, indicating effective reasoning and termination.

In contrast, prior reasoning-based methods such as PLR~\cite{abs-2601-03153} and ReaRec~\cite{abs-2503-22675} (including PRL and ERL) also report ceiling performance, but exhibit a substantially larger gap between their actual inference performance and the ceiling. This highlights ManCAR's ability to translate iterative refinement into near-optimal test-time
behavior, rather than relying on a fixed reasoning budget.

\subsection{Ablation Study (Tab.~\ref{tab:ablation_cds_video}).}

\begin{table}[t]
\centering
\caption{Ablation results on CDs and Video.}
\vspace{-10pt}
\label{tab:ablation_cds_video}
\scalebox{0.95}{
\begin{tabular}{lcccc}
\hline
\multirow{2}{*}{method} & \multicolumn{2}{c}{CDs} & \multicolumn{2}{c}{Video} \\
\cline{2-5}
 & N@10 & R@10 & N@10 & R@10 \\
\hline
ManCAR              & \textbf{0.0282} & \textbf{0.0516} & \textbf{0.0275} & \textbf{0.0521} \\
\hline
w/o Teacher Prior   & 0.0212 & 0.0436 & 0.0231 & 0.0445 \\

w/o Context         & 0.0234 & 0.0469 & 0.0265 & 0.0497 \\

w/o Norm Rescale    & 0.0277 & 0.0513 & 0.0271 & 0.0515 \\
\hline
w/o  Schedule & 0.0261 & 0.0496 & 0.0251 & 0.0491 \\

w/ Decrease Schedule & 0.0273 & 0.0490 & 0.0250 & 0.0518 \\
\hline

\end{tabular}
}
\vspace{-15pt}
\end{table}

\vspace{5pt}
\noindent\textbf{Graph-driven manifold constraint (w/o teacher prior)} results in the largest performance drop among ManCAR variants, though it still outperforms ContextBERT4Rec. This indicates that graph context alone provides limited gains, while the absence of teacher guidance makes target-driven reasoning susceptible to latent drift.

\vspace{5pt}
\noindent
\noindent\textbf{Context engineering (w/o context).}
Removing candidate-set context injection causes a clear performance drop, though this variant still outperforms ReaRec-style baselines. This suggests that teacher guidance alone can partially steer reasoning, while injecting graph-conditioned candidates as auxiliary context further narrows the predictive search space and improves target localization.

\vspace{5pt}
\noindent
\textbf{Latent state norm rescaling (w/o rescaling).}
Removing this module causes a consistent performance drop,
highlighting its role in aligning latent states with item embeddings. This normalization mitigates empirical norm growth and improves numerical stability, supporting stable stepwise refinement in multi-step reasoning.

\vspace{5pt}
\noindent
\noindent\textbf{Loss scheduling (w/o schedule or decreasing schedule).}
Removing the schedule or adopting a decreasing one in target prediction loss leads to clear performance degradation. This agrees with our analysis (Proposition~\ref{prop:continuation_tracking}) that conservative early updates helps to preserve manifold-walking stability and avoids premature convergence associated with decreasing schedules.

\subsection{Parameter Sensitivity}
\label{sec:parameter}

Fig.~\ref{fig:sensitivity} and~\ref{fig:ndcg_wrt_steps} present the parameter sensitivity of ManCAR with respect to \#context items (the number of graph neighbors used to construct the candidate set), \#training-steps (the step configuration used at training time), $\lambda$ (balancing target prediction and KL regularization), $\gamma_{\text{base}}$ (controlling the teacher sharpness schedule), and $\tau_base$ (controlling the temperature schedule for target prediction).

Among the major hyperparameters, ManCAR is most sensitive to the number of graph neighbors (construct the candidate set) and training-time steps, since either noise-injection or insufficient support and shaping of the manifold may degrade the performance.
In contrast, the model is relatively insensitive to the choice of $\lambda$ (balancing target prediction and KL regularization), $\gamma_{\text{base}}$ (controlling the teacher sharpness
schedule), and $\tau_1$ (controlling the temperature schedule for target prediction). Across these parameters, performance exhibits smooth and well-behaved trends, allowing optimal values to be reliably identified via simple grid search.

\begin{figure}[t]
  \centering
  \includegraphics[width=\columnwidth]{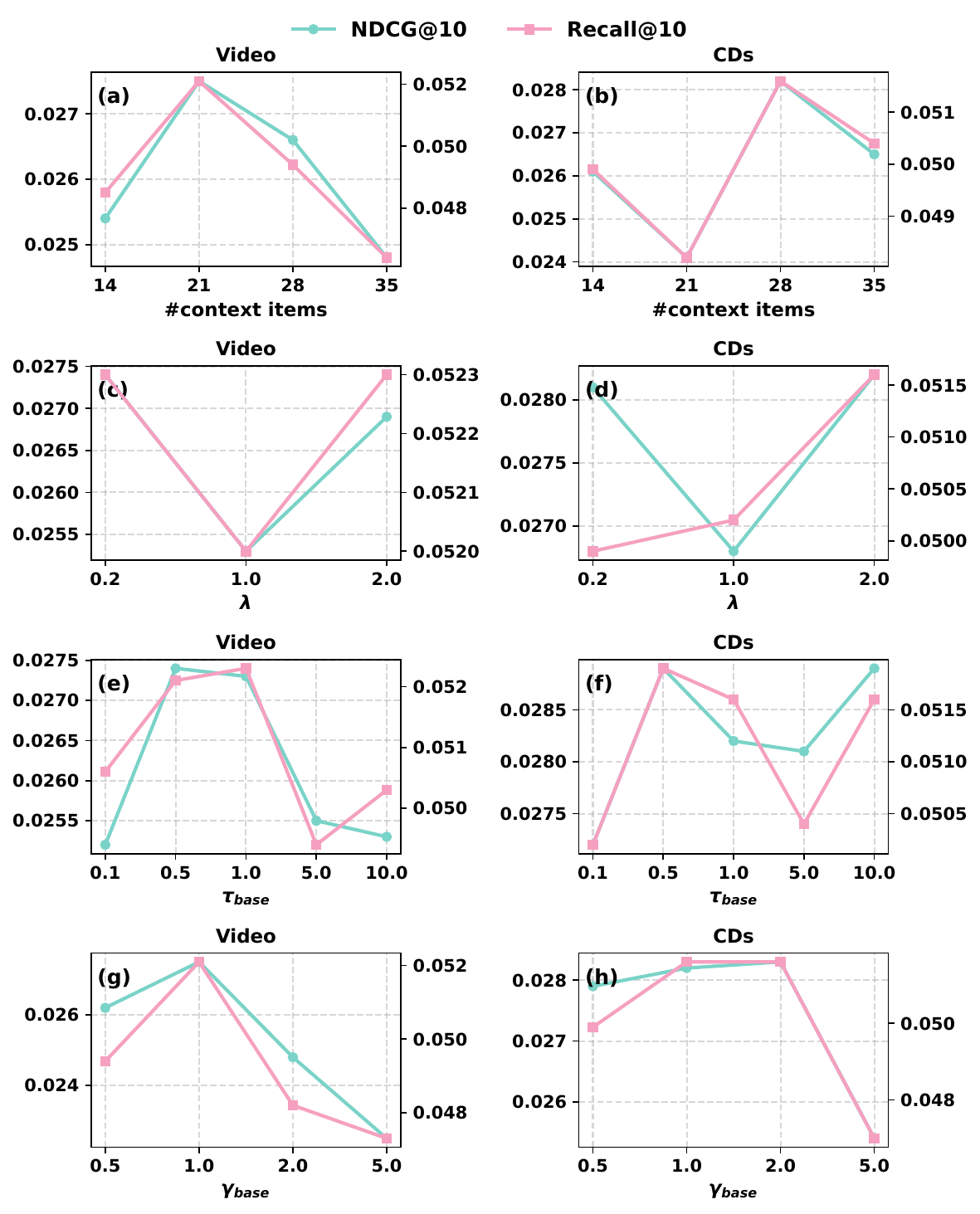}
  \caption{Sensitivity analysis on Video and CDs. (a) and (b): NDCG@10 and Recall@10 w.r.t. \#context items; (c) and (d): NDCG@10 and Recall@10 w.r.t. regularization loss weight $\lambda$; (e) and (f): NDCG@10 and Recall@10 w.r.t. temperature $\tau_{\mathrm{base}}$; (g) and (h): NDCG@10 and Recall@10 w.r.t. $\gamma_{\mathrm{base}}$.}
  \label{fig:sensitivity}
\end{figure} 

\begin{figure}[t]  
  \centering
  \includegraphics[width=\columnwidth]{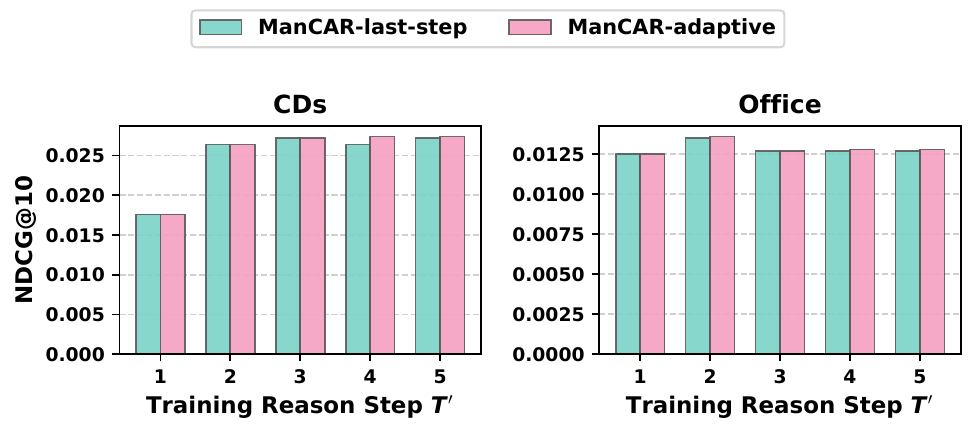}
  % \vspace{-20pt}
  \caption{NDCG@10 w.r.t. reason step $T^\prime$ on CDs and Office.}
  % \vspace{-5pt}
  \label{fig:ndcg_wrt_steps}
\end{figure}

\subsection{KL-Based Halting Analysis}
\label{sec:step_wise_kl}

We report the KL divergence between consecutive reasoning steps on two datasets, CDs and Video. For each test batch, we compute the average KL divergence across samples and then report the mean
and variance across batches. As shown in Fig.~\ref{fig:step_wise_kl}, after sufficient training, the KL divergence between adjacent reasoning steps decreases sharply, indicating stable convergence of the reasoning trajectory as expected.

\begin{figure}[t]
    \centering
    \includegraphics[width=\linewidth]{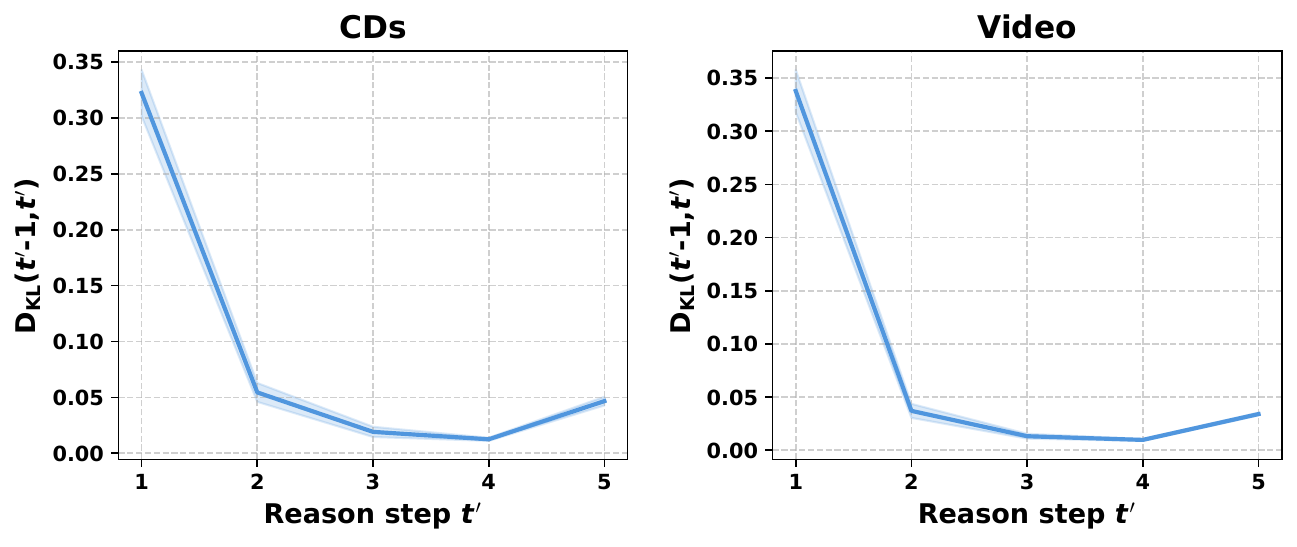}
    \caption{KL divergence between two adjacent step $t^\prime-1$ and $t$ w.r.t. inference steps $t^\prime$.}
    \label{fig:step_wise_kl}
\end{figure}

\subsection{Attention Visualization Analysis}
\label{sec:appx_atten_vis}

\begin{figure}[t]
    \centering
    \includegraphics[width=\linewidth]{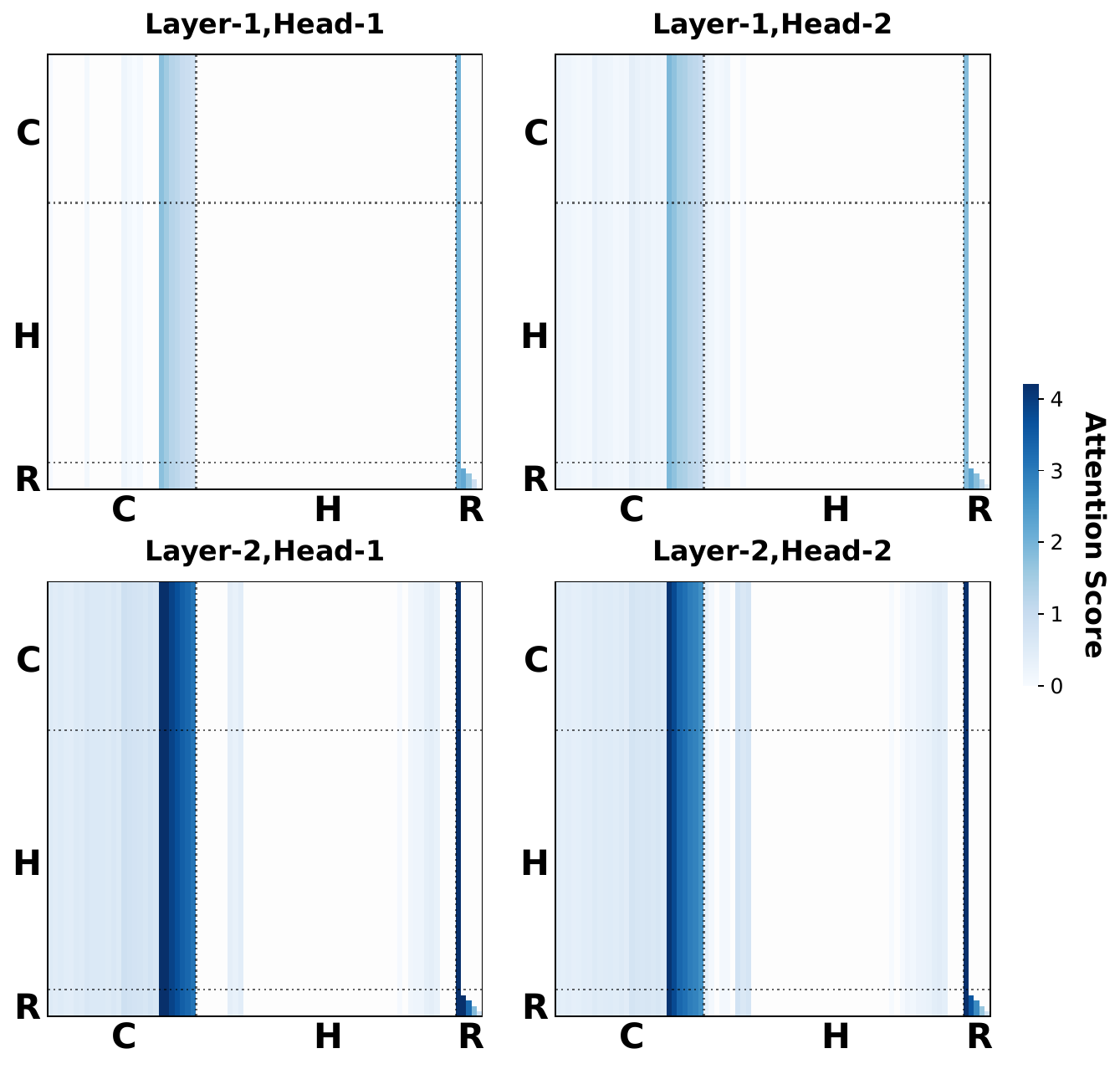}
    \caption{Attention Analysis of ManCAR on CDs. Attention scores are averaged over 1024 randomly sampled user histories from the test set.}
    \label{fig:attention_vis}
\end{figure}

The attention heatmaps in Fig.~\ref{fig:attention_vis} (two layers with two heads, with tokens partitioned into Context C, Interaction History H, and Reasoning Steps R) reveal a consistent routing pattern that aligns with ManCAR's manifold-constrained latent reasoning design. Across all heads, we observe a prominent concentration of attention mass from reasoning tokens toward a small subset of context positions (i.e., strong vertical bands within the $C$ region), while the H region is comparatively diffuse and weaker. This indicates that the intermediate reasoning states do not evolve in a free-form manner; instead, they repeatedly query the injected candidate context $\mathcal{C}(k)$ during refinement. 

Moreover, the deeper layer exhibits sharper and more structured attention: the $R \to C$ concentration becomes stronger, and we also see increased self-referential aggregation near the $R$ boundary (visible as emphasis close to the rightmost columns / bottom-right region). This suggests that later layers increasingly perform inter-step consolidation, integrating previous reasoning states while still grounding each update in the graph-conditioned candidate set. 

Additionally, recent user interactions—particularly the latest action—receive consistently larger attention scores, reflecting the recency bias commonly observed in practical recommender systems. Together, this indicates that ManCAR is building the data channel: Recent Action $\rightarrow$ Graph-Anchors (neighbors) $\rightarrow$ Reasoning States to achieve adaptive, stable, and constrained refinement within the local intent manifold.

\subsection{Additional Analyses}
For a comprehensive breakdown of the reasoning steps used during training and inference across all seven datasets, please refer to Appendix~\ref{sec:allocation}. Additionally, we provide a rigorous theoretical derivation of the computational complexity (FLOPs) for ManCAR and other baseline models in Appendix~\ref{sec:complexity}.
\section{Related Work}

\subsection{Sequential Recommendation}

As a core paradigm in recommendation, sequential recommendation captures user preferences to forecast the next item of interest. 

Non-LLM-Based sequential recommendation evolves from sequential pattern mining~\cite{YapLY12} and Markov chains~\cite{HeFWM16, HeM16} to recent deep learning approaches~\cite{KangM18,sun2019bert4rec}. 
Detailed surveys on non-LLM-Based sequential recommendation are available in~\cite{FangZSG20,WangHWCSO19}.

Recently, the emergence of LLMs has greatly affected the field of sequential recommendation, diverging into two paradigms~\cite{HuXZFWHLTZ24}: 
(1) LLM-Augmented sequential recommendation uses LLMs as feature extractors. LLMSeq~\cite{HarteZLKJF23} and SAID~\cite{HuXZFWHLTZ24} utilize LLM-derived embeddings for initialization and semantic alignment. 
Meanwhile, LRD~\cite{YangMSALCZ24} and SERALM~\cite{RenCYLJCZM024} leverage language knowledge to discover latent relations and refine generation via feedback from ID-based recommenders.
(2) LLM-Centric sequential recommendation employs the LLM as the predictor. 
Methods range from processing item ``sentences'' (RecFormer~\cite{li2023text}) and ID sequences (E4SRec~\cite{li2023e4srec}) to managing long sequences via summarization (LLM-TRSR~\cite{ZhengCQZ024}). 
Other works enhance reasoning through intent-driven prompting (LLM4ISR~\cite{0001LQFWO24}) and self-reflection agents (Re2LLM~\cite{abs-2403-16427}).

Besides, there is a burgeoning sequential recommendation paradigm called generative sequential recommendation~\cite{RajputMSKVHHT0S23,TanXHGLZ24,wang2024learnableitemtokenizationgenerative,ZhaiMWYZLT25} that replaces pre-fixed item IDs with identifiers constructed from generated tokens. 
By synthesizing tokens, these methods better leverage content to encode item semantics directly into the ID structure. 
However, this direction remains under-explored due to optimization challenges, such as the difficulty of distinguishing similar items with identical token sequences~\cite{ZhuJLQD024}.

\subsection{Reasoning-Enhanced Recommendation}

Reasoning-enhanced recommendation augments sequential recommendation with deliberative capabilities. 
It can be categorized into Explicit Reasoning (using visible, text-based chains) and Latent Reasoning (employing implicit, internal computation) to enhance recommendation accuracy. 

\vspace{5pt}
\noindent\textbf{Explicit Reasoning-Enhanced Recommendation.} 
Explicit reasoning approaches leverage the generative capabilities of LLMs to articulate the decision-making process through interpretable text or symbolic chains. 
R2ec~\cite{abs-2505-16994} introduces a unified dual-head architecture that simultaneously generates reasoning chains and predicts items. This design significantly reduces inference latency.
ReasoningRec~\cite{BismayDC25} bridges recommendations and explanations, and it uses CoT prompting to distill a LLM's synthetic reasoning into a smaller model. 
Reason4Rec~\cite{abs-2502-02061} formulates the deliberative recommendation task that incorporates explicit reasoning about user preferences as an alignment goal and enhances model's reasoning capabilities utilizing verbalized user feedback in a step-wise manner. 
Exp3rt~\cite{Kim0CKCY025} distills reasoning capabilities into a student LLM via a three-step process: preference extraction, profile construction, and prediction. It effectively utilizes rich review data for personalized recommendation. 
OneRec-Think~\cite{abs-2510-11639} introduces a ``Think-Ahead'' architecture that seamlessly integrates dialogue, reasoning, and personalized recommendation. 
RecGPT~\cite{abs-2512-14503} employs a Hierarchical Multi-Agent System for agentic intent reasoning and hybrid representation for efficiency, thereby solving the scalability issues of its predecessor, yet the complex multi-agent coordination introduces new challenges in system stability and debugging.

\vspace{5pt}
\noindent\textbf{Latent Reasoning-Enhanced Recommendation.} 
Inspired by latent reasoning for LLMs~\cite{abs-2412-06769}, recent sequential recommendation models have adopted latent reasoning to perform multi-step deliberation before prediction, without requiring explicit CoT data.  
ReaRec~\cite{abs-2503-22675} pioneers inference-time computing by autoregressively feeding the last hidden state back into the encoder to enhance performance. 
OnePiece~\cite{abs-2509-18091} applies latent reasoning to industrial retrieval and ranking by integrating context engineering with block-wise latent reasoning to progressively refine user intent. 
LARES~\cite{abs-2505-16865} employs depth-recurrent latent reasoning that leverages all the input tokens to perform multi-step reasoning. 
PLR~\cite{abs-2601-03153} introduces a width-level scaling paradigm that explores diverse reasoning paths simultaneously via parallel streams to alleviate diminishing returns as reasoning depth increases. 

\section{Conclusion}

We proposed \textbf{ManCAR}, a manifold-constrained latent reasoning framework for sequential recommendation.
By restricting latent refinement to a graph-locality-induced manifold and guiding it with progressive teacher supervision towards the target item, ManCAR enables stable and structured multi-step reasoning. A continuation-based analysis motivates both the teacher scheduling strategy and adaptive test-time termination.
Extensive experiments on seven public datasets demonstrate that ManCAR consistently outperforms strong sequential and reasoning-based baselines, yielding substantial improvements in retrieval and 
ranking quality. These results highlight the importance of explicit constraints over latent reasoning with concrete collaborative signals, and position ManCAR as a principled approach for controllable reasoning in sequential recommendation.

\bibliographystyle{ACM-Reference-Format}
\bibliography{ref}

@inproceedings{DeldjooHMKSRVSK24,
  author       = {Yashar Deldjoo and
                  Zhankui He and
                  Julian J. McAuley and
                  Anton Korikov and
                  Scott Sanner and
                  Arnau Ramisa and
                  Ren{\'{e}} Vidal and
                  Maheswaran Sathiamoorthy and
                  Atoosa Kasirzadeh and
                  Silvia Milano},
  title        = {A Review of Modern Recommender Systems Using Generative Models (Gen-RecSys)},
  booktitle    = {{KDD}},
  pages        = {6448--6458},
  year         = {2024}
}

@inproceedings{WuT0WXT19,
  author       = {Shu Wu and
                  Yuyuan Tang and
                  Yanqiao Zhu and
                  Liang Wang and
                  Xing Xie and
                  Tieniu Tan},
  title        = {Session-Based Recommendation with Graph Neural Networks},
  booktitle    = {{AAAI}},
  pages        = {346--353},
  year         = {2019}
}

@inproceedings{YingHCEHL18,
  author       = {Rex Ying and
                  Ruining He and
                  Kaifeng Chen and
                  Pong Eksombatchai and
                  William L. Hamilton and
                  Jure Leskovec},
  title        = {Graph Convolutional Neural Networks for Web-Scale Recommender Systems},
  booktitle    = {{KDD}},
  pages        = {974--983},
  year         = {2018}
}

@inproceedings{YapLY12,
  author       = {Ghim{-}Eng Yap and
                  Xiaoli Li and
                  Philip S. Yu},
  title        = {Effective Next-Items Recommendation via Personalized Sequential Pattern Mining},
  booktitle    = {{DASFAA}},
  volume       = {7239},
  pages        = {48--64},
  year         = {2012}
}

@inproceedings{HeFWM16,
  author       = {Ruining He and
                  Chen Fang and
                  Zhaowen Wang and
                  Julian J. McAuley},
  title        = {Vista: {A} Visually, Socially, and Temporally-aware Model for Artistic Recommendation},
  booktitle    = {RecSys},
  pages        = {309--316},
  year         = {2016}
}

@inproceedings{HeM16,
  author       = {Ruining He and
                  Julian J. McAuley},
  title        = {Fusing Similarity Models with Markov Chains for Sparse Sequential
                  Recommendation},
  booktitle    = {{ICDM}},
  pages        = {191--200},
  year         = {2016}
}

@inproceedings{KangM18,
  author       = {Wang{-}Cheng Kang and
                  Julian J. McAuley},
  title        = {Self-Attentive Sequential Recommendation},
  booktitle    = {{ICDM}},
  pages        = {197--206},
  year         = {2018}
}

@inproceedings{sun2019bert4rec,
  author       = {Fei Sun and
                  Jun Liu and
                  Jian Wu and
                  Changhua Pei and
                  Xiao Lin and
                  Wenwu Ou and
                  Peng Jiang},
  title        = {BERT4Rec: Sequential Recommendation with Bidirectional Encoder Representations from Transformer},
  booktitle    = {{CIKM}},
  pages        = {1441--1450},
  year         = {2019}
}

@article{FangZSG20,
  author       = {Hui Fang and
                  Danning Zhang and
                  Yiheng Shu and
                  Guibing Guo},
  title        = {Deep Learning for Sequential Recommendation: Algorithms, Influential Factors, and Evaluations},
  journal      = {{ACM} Trans. Inf. Syst.},
  volume       = {39},
  number       = {1},
  pages        = {10:1--10:42},
  year         = {2020}
}

@inproceedings{WangHWCSO19,
  author       = {Shoujin Wang and
                  Liang Hu and
                  Yan Wang and
                  Longbing Cao and
                  Quan Z. Sheng and
                  Mehmet A. Orgun},
  title        = {Sequential Recommender Systems: Challenges, Progress and Prospects},
  booktitle    = {{IJCAI}},
  pages        = {6332--6338},
  year         = {2019}
}

@inproceedings{HuXZFWHLTZ24,
  author       = {Jun Hu and
                  Wenwen Xia and
                  Xiaolu Zhang and
                  Chilin Fu and
                  Weichang Wu and
                  Zhaoxin Huan and
                  Ang Li and
                  Zuoli Tang and
                  Jun Zhou},
  title        = {Enhancing Sequential Recommendation via LLM-based Semantic Embedding Learning},
  booktitle    = {{WWW}},
  pages        = {103--111},
  year         = {2024}
}

@inproceedings{HarteZLKJF23,
  author       = {Jesse Harte and
                  Wouter Zorgdrager and
                  Panos Louridas and
                  Asterios Katsifodimos and
                  Dietmar Jannach and
                  Marios Fragkoulis},
  title        = {Leveraging Large Language Models for Sequential Recommendation},
  booktitle    = {RecSys},
  pages        = {1096--1102},
  year         = {2023}
}

@inproceedings{YangMSALCZ24,
  author       = {Shenghao Yang and
                  Weizhi Ma and
                  Peijie Sun and
                  Qingyao Ai and
                  Yiqun Liu and
                  Mingchen Cai and
                  Min Zhang},
  title        = {Sequential Recommendation with Latent Relations based on Large Language Model},
  booktitle    = {{SIGIR}},
  pages        = {335--344},
  year         = {2024}
}

@inproceedings{RenCYLJCZM024,
  author       = {Yankun Ren and
                  Zhongde Chen and
                  Xinxing Yang and
                  Longfei Li and
                  Cong Jiang and
                  Lei Cheng and
                  Bo Zhang and
                  Linjian Mo and
                  Jun Zhou},
  title        = {Enhancing Sequential Recommenders with Augmented Knowledge from Aligned Large Language Models},
  booktitle    = {{SIGIR}},
  pages        = {345--354},
  year         = {2024}
}

@inproceedings{li2023text,
  author       = {Jiacheng Li and
                  Ming Wang and
                  Jin Li and
                  Jinmiao Fu and
                  Xin Shen and
                  Jingbo Shang and
                  Julian J. McAuley},
  title        = {Text Is All You Need: Learning Language Representations for Sequential Recommendation},
  booktitle    = {{KDD}},
  pages        = {1258--1267},
  year         = {2023}
}

@article{li2023e4srec,
  author       = {Xinhang Li and
                  Chong Chen and
                  Xiangyu Zhao and
                  Yong Zhang and
                  Chunxiao Xing},
  title        = {E4SRec: An Elegant Effective Efficient Extensible Solution of Large Language Models for Sequential Recommendation},
  journal      = {arXiv Preprint},
  URL          = {https://arxiv.org/abs/2312.02443}, 
  year         = {2023}
}

@inproceedings{ZhengCQZ024,
  author       = {Zhi Zheng and
                  Wenshuo Chao and
                  Zhaopeng Qiu and
                  Hengshu Zhu and
                  Hui Xiong},
  title        = {Harnessing Large Language Models for Text-Rich Sequential Recommendation},
  booktitle    = {{WWW}},
  pages        = {3207--3216},
  year         = {2024}
}

@inproceedings{0001LQFWO24,
  author       = {Zhu Sun and
                  Hongyang Liu and
                  Xinghua Qu and
                  Kaidong Feng and
                  Yan Wang and
                  Yew Soon Ong},
  title        = {Large Language Models for Intent-Driven Session Recommendations},
  booktitle    = {{SIGIR}},
  pages        = {324--334},
  year         = {2024}
}

@article{abs-2403-16427,
  author       = {Ziyan Wang and
                  Yingpeng Du and
                  Zhu Sun and
                  Haoyan Chua and
                  Kaidong Feng and
                  Wenya Wang and
                  Jie Zhang},
  title        = {Re2LLM: Reflective Reinforcement Large Language Model for Session-based Recommendation},
  booktitle    = {{AAAI}},
  pages        = {12827--12835},
  year         = {2025}
}

@inproceedings{RajputMSKVHHT0S23,
  author       = {Shashank Rajput and
                  Nikhil Mehta and
                  Anima Singh and
                  Raghunandan Hulikal Keshavan and
                  Trung Vu and
                  Lukasz Heldt and
                  Lichan Hong and
                  Yi Tay and
                  Vinh Q. Tran and
                  Jonah Samost and
                  Maciej Kula and
                  Ed H. Chi and
                  Mahesh Sathiamoorthy},
  title        = {Recommender Systems with Generative Retrieval},
  booktitle    = {NeurIPS},
  pages        = {10299--10315},
  year         = {2023}
}

@inproceedings{TanXHGLZ24,
  author       = {Juntao Tan and
                  Shuyuan Xu and
                  Wenyue Hua and
                  Yingqiang Ge and
                  Zelong Li and
                  Yongfeng Zhang},
  title        = {IDGenRec: LLM-RecSys Alignment with Textual {ID} Learning},
  booktitle    = {{SIGIR}},
  pages        = {355--364},
  year         = {2024}
}

@inproceedings{wang2024learnableitemtokenizationgenerative,
  author       = {Wenjie Wang and
                  Honghui Bao and
                  Xinyu Lin and
                  Jizhi Zhang and
                  Yongqi Li and
                  Fuli Feng and
                  See{-}Kiong Ng and
                  Tat{-}Seng Chua},
  title        = {Learnable Item Tokenization for Generative Recommendation},
  booktitle    = {{CIKM}},
  pages        = {2400--2409},
  year         = {2024}
}

@inproceedings{ZhaiMWYZLT25,
  author       = {Jianyang Zhai and
                  Zi-Feng Mai and
                  Chang-Dong Wang and
                  Feidiao Yang and
                  Xiawu Zheng and
                  Hui Li and
                  Yonghong Tian},
  title        = {Multimodal Quantitative Language for Generative Recommendation},
  booktitle    = {{ICLR}},
  URL          = {https://openreview.net/forum?id=v7YrIjpkTF}, 
  year         = {2025}
}

@inproceedings{ZhuJLQD024,
  author       = {Jieming Zhu and
                  Mengqun Jin and
                  Qijiong Liu and
                  Zexuan Qiu and
                  Zhenhua Dong and
                  Xiu Li},
  title        = {CoST: Contrastive Quantization based Semantic Tokenization for Generative Recommendation},
  booktitle    = {RecSys},
  pages        = {969--974},
  year         = {2024}
}

@article{abs-2505-16994,
  author       = {Runyang You and
                  Yongqi Li and
                  Xinyu Lin and
                  Xin Zhang and
                  Wenjie Wang and
                  Wenjie Li and
                  Liqiang Nie},
  title        = {R\({}^{\mbox{2}}\)ec: Towards Large Recommender Models with Reasoning},
  journal      = {arXiv Preprint},
  url          = {https://arxiv.org/abs/2505.16994},
  year         = {2025}
}

@inproceedings{BismayDC25,
  author       = {Millennium Bismay and
                  Xiangjue Dong and
                  James Caverlee},
  title        = {ReasoningRec: Bridging Personalized Recommendations and Human-Interpretable Explanations through {LLM} Reasoning},
  booktitle    = {{NAACL} (Findings)},
  pages        = {8132--8148},
  year         = {2025}
}

@article{abs-2502-02061,
  author       = {Yi Fang and
                  Wenjie Wang and
                  Yang Zhang and
                  Fengbin Zhu and
                  Qifan Wang and
                  Fuli Feng and
                  Xiangnan He},
  title        = {Reason4Rec: Large Language Models for Recommendation with Deliberative User Preference Alignment},
  journal      = {arXiv Preprint},
  url          = {https://arxiv.org/abs/2502.02061},
  year         = {2025}
}

@inproceedings{Kim0CKCY025,
  author       = {Jieyong Kim and
                  Hyunseo Kim and
                  Hyunjin Cho and
                  SeongKu Kang and
                  Buru Chang and
                  Jinyoung Yeo and
                  Dongha Lee},
  title        = {Review-driven Personalized Preference Reasoning with Large Language Models for Recommendation},
  booktitle    = {{SIGIR}},
  pages        = {1697--1706},
  year         = {2025}
}

@article{abs-2510-11639,
  author       = {Zhanyu Liu and
                  Shiyao Wang and
                  Xingmei Wang and
                  Rongzhou Zhang and
                  Jiaxin Deng and
                  Honghui Bao and
                  Jinghao Zhang and
                  Wuchao Li and
                  Pengfei Zheng and
                  Xiangyu Wu and
                  Yifei Hu and
                  Qigen Hu and
                  Xinchen Luo and
                  Lejian Ren and
                  Zixing Zhang and
                  Qianqian Wang and
                  Kuo Cai and
                  Yunfan Wu and
                  Hongtao Cheng and
                  Zexuan Cheng and
                  Lu Ren and
                  Huanjie Wang and
                  Yi Su and
                  Ruiming Tang and
                  Kun Gai and
                  Guorui Zhou},
  title        = {OneRec-Think: In-Text Reasoning for Generative Recommendation},
  journal      = {arXiv Preprint},
  url          = {https://arxiv.org/abs/2510.11639},
  year         = {2025}
}

@article{abs-2512-14503,
  author       = {Chao Yi and
                  Dian Chen and
                  Gaoyang Guo and
                  Jiakai Tang and
                  Jian Wu and
                  Jing Yu and
                  Mao Zhang and
                  Wen Chen and
                  Wenjun Yang and
                  Yujie Luo and
                  Yuning Jiang and
                  Zhujin Gao and
                  Bo Zheng and
                  Binbin Cao and
                  Changfa Wu and
                  Dixuan Wang and
                  Han Wu and
                  Haoyi Hu and
                  Kewei Zhu and
                  Lang Tian and
                  Lin Yang and
                  Qiqi Huang and
                  Siqi Yang and
                  Wenbo Su and
                  Xiaoxiao He and
                  Xin Tong and
                  Xu Chen and
                  Xunke Xi and
                  Xiaowei Huang and
                  Yaxuan Wu and
                  Yeqiu Yang and
                  Yi Hu and
                  Yujin Yuan and
                  Yuliang Yan and
                  Zile Zhou},
  title        = {RecGPT-V2 Technical Report},
  journal      = {arXiv Preprint},
  url          = {https://arxiv.org/abs/2512.14503},
  year         = {2025}
}

@article{abs-2412-06769,
  author       = {Shibo Hao and
                  Sainbayar Sukhbaatar and
                  DiJia Su and
                  Xian Li and
                  Zhiting Hu and
                  Jason Weston and
                  Yuandong Tian},
  title        = {Training Large Language Models to Reason in a Continuous Latent Space},
  journal      = {arXiv Preprint},
  url          = {https://arxiv.org/abs/2412.06769},
  year         = {2024}
}

@article{abs-2503-22675,
  author       = {Jiakai Tang and
                  Sunhao Dai and
                  Teng Shi and
                  Jun Xu and
                  Xu Chen and
                  Wen Chen and
                  Wu Jian and
                  Yuning Jiang},
  title        = {Think Before Recommend: Unleashing the Latent Reasoning Power for Sequential Recommendation},
  journal      = {arXiv Preprint},
  url          = {https://arxiv.org/abs/2503.22675},
  year         = {2025}
}

@article{abs-2505-16865,
  author       = {Enze Liu and
                  Bowen Zheng and
                  Xiaolei Wang and
                  Wayne Xin Zhao and
                  Jinpeng Wang and
                  Sheng Chen and
                  Ji{-}Rong Wen},
  title        = {{LARES:} Latent Reasoning for Sequential Recommendation},
  journal      = {arXiv Preprint},
  url          = {https://arxiv.org/abs/2505.16865},
  year         = {2025}
}

@article{abs-2601-03153,
  author       = {Jiakai Tang and 
                  Xu Chen and 
                  Wen Chen and 
                  Jian Wu and 
                  Yuning Jiang and 
                  Bo Zheng},
  title        = {Parallel Latent Reasoning for Sequential Recommendation},
  journal      = {arXiv Preprint},
  url          = {https://arxiv.org/abs/2601.03153},
  year         = {2026}
}

@article{abs-2509-18091,
  author       = {Sunhao Dai and
                  Jiakai Tang and
                  Jiahua Wu and
                  Kun Wang and
                  Yuxuan Zhu and
                  Bingjun Chen and
                  Bangyang Hong and
                  Yu Zhao and
                  Cong Fu and
                  Kangle Wu and
                  Yabo Ni and
                  Anxiang Zeng and
                  Wenjie Wang and
                  Xu Chen and
                  Jun Xu and
                  See{-}Kiong Ng},
  title        = {OnePiece: Bringing Context Engineering and Reasoning to Industrial Cascade Ranking System},
  journal      = {arXiv Preprint},
  url          = {https://arxiv.org/abs/2509.18091},
  year         = {2025}
}

@article{abs-2507-13334,
  author       = {Lingrui Mei and
                  Jiayu Yao and
                  Yuyao Ge and
                  Yiwei Wang and
                  Baolong Bi and
                  Yujun Cai and
                  Jiazhi Liu and
                  Mingyu Li and
                  Zhong{-}Zhi Li and
                  Duzhen Zhang and
                  Chenlin Zhou and
                  Jiayi Mao and
                  Tianze Xia and
                  Jiafeng Guo and
                  Shenghua Liu},
  journal      = {arXiv Preprint},
  url          = {https://arxiv.org/abs/2507.13334},
  year         = {2025}
}

@article{abs-2010-05525,
  author       = {Xiaoyong Yang and
                  Yadong Zhu and
                  Yi Zhang and
                  Xiaobo Wang and
                  Quan Yuan},
  title        = {Large Scale Product Graph Construction for Recommendation in E-commerce},
  journal      = {arXiv Preprint},
  url          = {https://arxiv.org/abs/2010.05525},
  year         = {2020}
}

@misc{layernormscaling,
      author={Wenfang Sun and Xinyuan Song and Pengxiang Li and Lu Yin and Yefeng Zheng and Shiwei Liu},
      title={The Curse of Depth in Large Language Models}, 
      journal  = {arXiv Preprint},
      url={https://arxiv.org/abs/2502.05795}, 
      year={2025}
}

@inproceedings{ShenYZHDH25,
  author       = {Zhenyi Shen and
                  Hanqi Yan and
                  Linhai Zhang and
                  Zhanghao Hu and
                  Yali Du and
                  Yulan He},
  title        = {{CODI:} Compressing Chain-of-Thought into Continuous Space via Self-Distillation},
  booktitle    = {{EMNLP}},
  pages        = {677--693},
  year         = {2025}
}

@article{abs-2509-02350,
  author       = {Jindong Li and
                  Yali Fu and
                  Li Fan and
                  Jiahong Liu and
                  Yao Shu and
                  Chengwei Qin and
                  Menglin Yang and
                  Irwin King and
                  Rex Ying},
  title        = {Implicit Reasoning in Large Language Models: {A} Comprehensive Survey},
  journal      = {arXiv Preprint},
  url={https://arxiv.org/abs/2509.02350},
  year         = {2025}
}

@article{abs-2510-15522,
  author       = {Jingcheng Deng and
                  Liang Pang and
                  Zihao Wei and
                  Shicheng Xu and
                  Zenghao Duan and
                  Kun Xu and
                  Yang Song and
                  Huawei Shen and
                  Xueqi Cheng},
  title        = {Latent Reasoning in LLMs as a Vocabulary-Space Superposition},
  journal      = {arXiv Preprint},
  url          = {https://arxiv.org/abs/2510.15522},
  year         = {2025}
}

@inproceedings{YuY00CN22,
  author       = {Junliang Yu and
                  Hongzhi Yin and
                  Xin Xia and
                  Tong Chen and
                  Lizhen Cui and
                  Quoc Viet Hung Nguyen},
  title        = {Are Graph Augmentations Necessary?: Simple Graph Contrastive Learning
                  for Recommendation},
  booktitle    = {{SIGIR}},
  pages        = {1294--1303},
  year         = {2022}
}

@article{JuFGLLQQSSXYYZWLZ24,
  author       = {Wei Ju and
                  Zheng Fang and
                  Yiyang Gu and
                  Zequn Liu and
                  Qingqing Long and
                  Ziyue Qiao and
                  Yifang Qin and
                  Jianhao Shen and
                  Fang Sun and
                  Zhiping Xiao and
                  Junwei Yang and
                  Jingyang Yuan and
                  Yusheng Zhao and
                  Yifan Wang and
                  Xiao Luo and
                  Ming Zhang},
  title        = {A Comprehensive Survey on Deep Graph Representation Learning},
  journal      = {Neural Networks},
  volume       = {173},
  pages        = {106207},
  year         = {2024}
}

@inproceedings{YangWWHHW24,
  author       = {Yonghui Yang and
                  Le Wu and
                  Zihan Wang and
                  Zhuangzhuang He and
                  Richang Hong and
                  Meng Wang},
  title        = {Graph Bottlenecked Social Recommendation},
  booktitle    = {{KDD}},
  pages        = {3853--3862},
  year         = {2024}
}

@inproceedings{YangHXL22,
  author       = {Yuhao Yang and
                  Chao Huang and
                  Lianghao Xia and
                  Chenliang Li},
  title        = {Knowledge Graph Contrastive Learning for Recommendation},
  booktitle    = {{SIGIR}},
  pages        = {1434--1443},
  year         = {2022}
}

@inproceedings{ChangGZHNSJ021,
  author       = {Jianxin Chang and
                  Chen Gao and
                  Yu Zheng and
                  Yiqun Hui and
                  Yanan Niu and
                  Yang Song and
                  Depeng Jin and
                  Yong Li},
  title        = {Sequential Recommendation with Graph Neural Networks},
  booktitle    = {{SIGIR}},
  pages        = {378--387},
  year         = {2021}
}

@inproceedings{WeiRTWSCWYH24,
  author       = {Wei Wei and
                  Xubin Ren and
                  Jiabin Tang and
                  Qinyong Wang and
                  Lixin Su and
                  Suqi Cheng and
                  Junfeng Wang and
                  Dawei Yin and
                  Chao Huang},
  title        = {LLMRec: Large Language Models with Graph Augmentation for Recommendation},
  booktitle    = {{WSDM}},
  pages        = {806--815},
  year         = {2024}
}

@article{abs-2403-03952,
  author       = {Yupeng Hou and
                  Jiacheng Li and
                  Zhankui He and
                  An Yan and
                  Xiusi Chen and
                  Julian J. McAuley},
  title        = {Bridging Language and Items for Retrieval and Recommendation},
  journal      = {arXiv Preprint},
  url          = {https://arxiv.org/abs/2403.03952},
  year         = {2024}
}
\clearpage
\appendix
\section{Derivation of Eq.~\ref{equation:elbo-obj}}
\label{sec:ELBO-derive}

\begin{proposition}[Graph-Conditioned Variational Regularization Training Objective]
\label{prop:graph_elbo}
Let $H$ denote a user interaction history and $i^*$ the ground-truth next item observed at training time.  
Let $\mathcal{C}(k)$ be the candidate set induced by the $k$-hop neighborhood of the most recent items $I_n$ on the interaction graph $\mathcal{G}$.  

Consider the latent-variable formulation
\[
p_\theta(i^* \mid H)
=
\sum_{c \in \mathcal{C}(k)}
p_\theta(c \mid H)\, p_\theta(i^* \mid c, H),
\]
where $c$ is a discrete latent intent prototype.  
Let $q(c \mid I_n, \mathcal{G})$ be any categorical distribution supported on $\mathcal{C}(k)$ that does not depend on $\theta$.  
Then, for all $\theta$, the following inequality holds:
\begin{align}
\log p_\theta(i^* \mid H)
\;\ge\;
&\mathbb{E}_{q(c \mid I_n, \mathcal{G})}
\big[ \log p_\theta(i^* \mid c, H) \big] \nonumber\\
&-
D_{\mathrm{KL}}\!\left(
q(c \mid I_n, \mathcal{G})
\;\|\;
p_\theta(c \mid H)
\right). \nonumber
\end{align}
The right-hand side defines an ELBO-like objective that regularizes the inferred intent distribution toward the graph-conditioned prior.
\end{proposition}

\begin{proof}
Starting from the marginal likelihood,
\[
p_\theta(i^* \mid H)
=
\sum_{c \in \mathcal{C}(k)}
p_\theta(c \mid H)\, p_\theta(i^* \mid c, H).
\]
For any categorical distribution $q(c \mid I_n, \mathcal{G})$ supported on $\mathcal{C}(k)$, we rewrite the sum as an expectation:
\begin{align}
p_\theta(i^* \mid H)
=
&\sum_{c \in \mathcal{C}(k)}
q(c \mid I_n, \mathcal{G})
\frac{p_\theta(c \mid H)\, p_\theta(i^* \mid c, H)}
     {q(c \mid I_n, \mathcal{G})} \nonumber\\
=
&\mathbb{E}_{q}\!\left[
\frac{p_\theta(c \mid H)\, p_\theta(i^* \mid c, H)}
     {q(c \mid I_n, \mathcal{G})}
\right]. \nonumber
\end{align}

Taking logarithm and applying Jensen's inequality yields
\begin{align}
\log p_\theta(i^* \mid H)
\;\ge\;
\mathbb{E}_{q}\!\Big[
&\log p_\theta(c \mid H)
+
\log p_\theta(i^* \mid c, H) \nonumber\\
&-
\log q(c \mid I_n, \mathcal{G})
\Big]. \nonumber
\end{align}

Rearranging terms gives
\begin{align}
\log p_\theta(i^* \mid H)
\;\ge\;
&\mathbb{E}_{q}\!\left[
\log p_\theta(i^* \mid c, H)
\right] \nonumber\\
&-
D_{\mathrm{KL}}\!\left(
q(c \mid I_n, \mathcal{G})
\;\|\;
p_\theta(c \mid H)
\right), \nonumber
\end{align} 
which completes the proof.
\end{proof}

\section{Proof of Proposition~\ref{prop:local_smoothness}}
\label{sec:proof_local_smoothness}

\begin{proof}
For a fixed candidate set $\mathcal{C}$, define
\begin{align}
P(c \mid H)
&= \frac{\exp\!\left(\mathbf{r}^\top \mathbf{e}_c\right)}
{\sum_{c' \in \mathcal{C}} \exp\!\left(\mathbf{r}^\top \mathbf{e}_{c'}\right)} \nonumber\\
&= \frac{\exp\!\left(\mathbf{r}^\top \mathbf{e}_c\right)}{Z(\mathbf{r})},\nonumber\\
Z(\mathbf{r}) &:= \sum_{c' \in \mathcal{C}} \exp\!\left(\mathbf{r}^\top \mathbf{e}_{c'}\right),\nonumber
\end{align}
where we suppress the dependence on $H$ in $\mathbf{r}$ for notational simplicity.

The KL distillation loss is
\begin{align}
\mathcal{L}(\mathbf{r})
&= \mathrm{D}_{KL}\!\left(Q \,\|\, P(\cdot \mid H)\right) \nonumber\\
&= \sum_{c \in \mathcal{C}} Q(c)\log\frac{Q(c)}{P(c\mid H)} \nonumber\\
&= \sum_{c \in \mathcal{C}} Q(c)\log Q(c)
\;-\;
\sum_{c \in \mathcal{C}} Q(c)\log P(c\mid H).\nonumber
\end{align}

The first term is constant with respect to $r$, hence
\begin{equation}
\nabla_{\mathbf{r}} \mathcal{L}(\mathbf{r})
= - \sum_{c \in \mathcal{C}} Q(c)\, \nabla_{\mathbf{r}} \log P(c\mid H).\nonumber
\end{equation}

Next, using $\log P(c\mid H)= \mathbf{r}^\top \mathbf{e}_c - \log Z(\mathbf{r})$, we have
\begin{equation}
\nabla_{\mathbf{r}} \log P(c\mid H)
= \nabla_{\mathbf{r}} (\mathbf{r}^\top \mathbf{e}_c) - \nabla_{\mathbf{r}} \log Z(\mathbf{r})
= \mathbf{e}_c - \frac{1}{Z(\mathbf{r})} \nabla_{\mathbf{r}} Z(\mathbf{r}).\nonumber
\end{equation}
Moreover,
\begin{equation}
\nabla_{\mathbf{r}} Z(\mathbf{r})
= \nabla_{\mathbf{r}} \sum_{c' \in \mathcal{C}} \exp\!\left(\mathbf{r}^\top \mathbf{e}_{c'}\right)
= \sum_{c' \in \mathcal{C}} \exp\!\left(\mathbf{r}^\top \mathbf{e}_{c'}\right) \mathbf{e}_{c'}.\nonumber
\end{equation}
Substituting back yields
\begin{align}
\nabla_{\mathbf{r}} \log P(c\mid H)
&= \mathbf{e}_c - \frac{1}{Z(\mathbf{r})} \sum_{c' \in \mathcal{C}} \exp\!\left(\mathbf{r}^\top \mathbf{e}_{c'}\right) \mathbf{e}_{c'} \nonumber\\
&= \mathbf{e}_c - \sum_{c' \in \mathcal{C}} \frac{\exp\!\left(\mathbf{r}^\top \mathbf{e}_{c'}\right)}{Z(\mathbf{r})} \mathbf{e}_{c'} \nonumber\\
&= \mathbf{e}_c - \sum_{c' \in \mathcal{C}} P(c' \mid H)\, \mathbf{e}_{c'} \nonumber\\
&= \mathbf{e}_c - \mathbb{E}_{P(\cdot \mid H)}[\mathbf{e}_c]. \nonumber
\end{align}

Therefore,
\begin{align}
\nabla_{\mathbf{r}} \mathcal{L}(\mathbf{r})
&= - \sum_{c \in \mathcal{C}} Q(c)\left(\mathbf{e}_c - \mathbb{E}_{P(\cdot \mid H)}[\mathbf{e}_c]\right) \nonumber\\
&= - \sum_{c \in \mathcal{C}} Q(c)\mathbf{e}_c \;+\; \left(\sum_{c \in \mathcal{C}} Q(c)\right)\mathbb{E}_{P(\cdot \mid H)}[\mathbf{e}_c].\nonumber
\end{align}
Since $Q$ is a probability distribution on $\mathcal{C}$, $\sum_{c \in \mathcal{C}} Q(c)=1$, and thus
\begin{equation}
\nabla_{\mathbf{r}} \mathcal{L}(\mathbf{r})
= \mathbb{E}_{P(\cdot \mid H)}[\mathbf{e}_c] - \mathbb{E}_{Q}[\mathbf{e}_c],\nonumber
\end{equation}
which completes the proof.
\end{proof}

\section{Proof of Proposition~\ref{prop:continuation_tracking}}
\label{sec:proof_continuation_tracking}

\begin{proof}
By the triangle inequality of the total variation distance, we have
\begin{equation}
d_{\mathrm{TV}}(p_{t'+1}, q_{t'+1})
\le
d_{\mathrm{TV}}(p_{t'+1}, q_{t'})
+
d_{\mathrm{TV}}(q_{t'}, q_{t'+1}).\nonumber
\label{eq:tv_triangle}
\end{equation}
Applying the bounded teacher drift assumption in Eq.~\ref{eq:assump_drift} yields
\begin{equation}
d_{\mathrm{TV}}(p_{t'+1}, q_{t'+1})
\le
d_{\mathrm{TV}}(p_{t'+1}, q_{t'}) + \delta.\nonumber
\label{eq:tv_drift}
\end{equation}
Using the stepwise contraction assumption in Eq.~\ref{eq:assump_contraction}, we obtain the recursive bound:
\begin{equation}
d_{\mathrm{TV}}(p_{t'+1}, q_{t'+1})
\le
(1-\lambda)\, d_{\mathrm{TV}}(p_{t'}, q_{t'}) + \delta.\nonumber
\label{eq:tv_recursion}
\end{equation}
Unrolling the recursion for $t'-1$ steps gives
\begin{align}
d_{\mathrm{TV}}(p_{t'}, q_{t'})
&\le
(1-\lambda)^{t'-1}\, d_{\mathrm{TV}}(p_1, q_1)
+
\delta \sum_{j=0}^{t'-2} (1-\lambda)^j
\nonumber\\
&=
(1-\lambda)^{t'-1}\, d_{\mathrm{TV}}(p_1, q_1)
+
\delta \cdot \frac{1-(1-\lambda)^{t'-1}}{\lambda}
\nonumber\\
&\le
(1-\lambda)^{t'-1}\, d_{\mathrm{TV}}(p_1, q_1)
+
\frac{\delta}{\lambda},\nonumber
\end{align}
which establishes the desired bound in Eq.~\ref{eq:tracking_bound}.
\end{proof}

\section{Global Relation Modeling via Swing Graph}
\label{sec:swing_build}
To capture stable collaborative signals, we construct a global item graph $\mathcal{G} = (\mathcal{V}, \mathcal{E})$ using an enhanced variant of the Swing algorithm used in industrial practice (e.g., Alibaba~\cite{abs-2010-05525}). This variant incorporates user activity normalization and popularity smoothing to mitigate the impact of noise from hyper-active users and hot items.

Formally, let $\mathcal{U}_i$ denote the set of users who interacted with item $i$, and $I_u$ denote the interaction history of user $u$. For a pair of items $(i, j)$, we first identify the set of common users $\mathcal{K}_{ij} = \mathcal{U}_i \cap \mathcal{U}_j$. To improve efficiency, if $|\mathcal{K}_{ij}|$ exceeds a threshold $M$, we perform random sampling to obtain a subset $\hat{\mathcal{K}}_{ij} \subset \mathcal{K}_{ij}$.

The similarity score $Sim(i, j)$ is defined as a weighted summation over user pairs $(u, v)$ from $\hat{\mathcal{K}}_{ij}$:

\begin{equation}
    Sim(i, j) = \frac{1}{\sqrt{|\mathcal{U}_j|}} \sum_{u \in \hat{\mathcal{K}}_{ij}} \sum_{v \in \hat{\mathcal{K}}_{ij}, v \neq u} w_{uv},\nonumber
\end{equation}
where the pair weight $w_{uv}$ combines user activity decay and substructure strength:

\begin{equation}
    w_{uv} = \underbrace{\frac{1}{(|I_u| + \alpha_1)^\beta \cdot (|I_v| + \alpha_1)^\beta}}_{\text{User Activity Weight}} \cdot \underbrace{\frac{1}{|I_u \cap I_v| + \alpha_2}}_{\text{Overlap Penalty}} \nonumber
\end{equation}
where $\alpha_1, \alpha_2$ are smoothing parameters, and $\beta$ controls the strength of user activity penalization. The term $1/\sqrt{|\mathcal{U}_j|}$ acts as a normalization factor to prevent popular items from dominating the retrieval results. This formulation ensures that the ``Intent Anchors'' are derived from high-quality, non-trivial collaborative structures.

\section{ManCAR Algorithms}
\label{sec:app_mancar_algo}
We summarize the implementations of ManCAR's training and adaptive reasoning in Algorithms~\ref{alg:man_car_train} and \ref{alg:man_car_infer}, respectively. 

Notably, we find that $k=1$ is sufficient for all Amazon Reviews datasets and adopt this setting throughout our experiments.
Increasing $k$ consistently degrades performance on these datasets, likely due to their high sparsity induced by the dataset construction process. In such settings, expanding to higher-order neighborhoods introduces noisy candidates that dilute useful collaborative signals. We note that larger $k$ may be beneficial in denser, real-world industrial scenarios, where higher-order relations are more reliable, and we recommend applying exponential decay to $k$-hop neighbor weights to mitigate noise when increasing $k$.

\begin{algorithm}[t]
\caption{ManCAR Training Algorithm}\label{alg:man_car_train}
\begin{algorithmic}[1]
\Require train set $\{H^{(j)},i^{*(j)}\}_{j=1}^{N}$, batch size $B$, \#hops $k$, recent interaction window size $n$, reasoning step $T$

\State Construct global interaction graph $\mathcal{G}\gets \mathrm{SWING}(\{H^{(j)}\}_{j=1}^{N})$

\For{randomly sampled mini-batch $\{H^{(j)},i^{*(j)}\}_{j=1}^{B}$}

\State $\mathcal{C}^{(j)}(k) \gets \mathcal{C}(I_n^{(j)};\mathcal{G};k)$

\For{$t=1,\dots T$}

\State Get teacher prior $q^{(t)}(c^{(j)}|I^{(j)}_n,\mathcal{G})$ via Eq.~\ref{eqn:q_strategy_2}

\State Compute the main loss $\mathcal{L}^{(t)}_{\mathrm{main}}$ via Eq.~\ref{eqn:loss_main_t}

\State Compute the regularization loss $\mathcal{L}^{(t)}_{\mathrm{reg}}$ via Eq.~\ref{eqn:loss_kl_t}

\EndFor

\State Compute the overall loss $\mathcal{L}$ via Eq.~\ref{eqn:loss_overall}

\State Minimize $\mathcal{L}$ w.r.t. $\theta$

\EndFor
\State return $\theta$
\end{algorithmic}
\end{algorithm}

\begin{algorithm}[t]
\caption{ManCAR Adaptive Reasoning Inference Algorithm}\label{alg:man_car_infer}
\begin{algorithmic}[1]
\Require inference input $H$, global interaction graph $\mathcal{G}$ constructed on train set, \#hops $k$, recent interaction window size $n$, max reasoning step $T_{\mathrm{max}}$, early stop threshold $\epsilon$

\State Define $p_{\theta}^{(0)}\gets \mathrm{NULL}$

\For{$t=1\dots T_{\mathrm{max}}$}

\State $\mathcal{C}(k) \gets \mathcal{C}(I_n;\mathcal{G};k)$

\State $p_{\theta}^{(t)}\gets f_{\theta}^{(t)}(H,\mathcal{C}(k))$
\If{$p_{\theta}^{(t-1)}$ is not $\mathrm{NULL}$}

\If{$D_{\mathrm{KL}}(p_{\theta}^{(t-1)}||p_{\theta}^{(t)})<\epsilon$}

\State End the reasoning at step $t$

\EndIf

\EndIf

\EndFor
\State return $p_{\theta}^{(t)}$

\end{algorithmic}
\end{algorithm}

\section{Additional Analyses}

\subsection{More Results for Data-Aware Train-Test Compute Allocation}
\label{sec:allocation}

Tab.~\ref{tab:step_full} reports the reason step in train and inference phase of ERL, PRL, PLR, and ManCAR on seven datasets, providing supplementary results to Tab.~\ref{tab:step_partial} of Sec.~\ref{sec:in-depthana}.

\begin{table}[t]
    \centering
    \caption{Number of reasoning steps used during training and inference for ERL, PRL, PLR, and ManCAR across seven datasets. Note that LARES adopts a loop-architecture-based method rather than forward step-wise reasoning; the reported step count for LARES corresponds to the number of loop iterations.}

    \vspace{-10pt}
    \resizebox{\linewidth}{!}{
    \begin{tabular}{ll|cccc|c}
    \hline
    Dataset                  & Reason step & ERL & PRL & PLR & LARES & ManCAR \\ \hline
    \multirow{2}{*}{CDs}     & Train step  & 2   & 2   & 3  &4  & 5      \\
                             & Infer step  & 2   & 2   & 3  &4  & 1.84   \\ \hline
    \multirow{2}{*}{Video}   & Train step  & 2   & 2   & 2  &4 & 4      \\
                             & Infer step  & 2   & 2   & 2  &4 & 3.68   \\ \hline
    \multirow{2}{*}{Office}  & Train step  & 2   & 2   & 2  &4 & 2      \\
                             & Infer step  & 2   & 2   & 2  &4 & 1.93   \\ \hline
    \multirow{2}{*}{Arts}    & Train step  & 2   & 2   & 1  &4 & 1      \\
                             & Infer step  & 2   & 2   & 1  &4 & 1      \\ \hline
    \multirow{2}{*}{Music}   & Train step  & 2   & 2   & 2  &4 & 2      \\
                             & Infer step  & 2   & 2   & 2  &4 & 1.84   \\ \hline
    \multirow{2}{*}{Toys}    & Train step  & 1   & 2   & 1  &4 & 4      \\
                             & Infer step  & 1   & 2   & 1  &4 & 3.58   \\ \hline
    \multirow{2}{*}{Grocery} & Train step  & 2   & 2   & 3  &4 & 2      \\
                             & Infer step  & 2   & 2   & 3  &4 & 1.74   \\ \hline
    \end{tabular}
    }
\label{tab:step_full}
\end{table}

\subsection{Computation Complexity Analysis}
\label{sec:complexity}

Let $|C|$ and $|H|$ denote the length of context prompt and user history, respectively, $d$ denote the hidden dimension, $L$ denote the number of Transformer layers, $T'$ denote the number of reasoning steps.

\vspace{5pt}
\noindent
\textbf{Transformer Encoder.}
In multi-head self attention, the Q/K/V/out projection and weighted sum costs $\mathcal{O}((|C|+|H|)d^2)$ and $\mathcal{O}((|C|+|H|)^2d)$. In FFN, the two linear-layer costs $\mathcal{O}((|C|+|H|)d^2)$. Thus, the total FLOPs for $L$-layer Transformer encoder is $\mathcal{O}(L((|C|+|H|)^2d+(|C|+|H|)d^2))$.

\vspace{5pt}
\noindent
\textbf{Autoregressive Reasoning.}
With KV cache enabled, for each step $t'\in[1,...,T']$, the cost of Q/K/V/out projection is $\mathcal{O}(d^2)$, the cost of attention weighted sum is $\mathcal{O}((|C|+|H|+t'-1)d)$, the cost of FFN is $\mathcal{O}(d^2)$. Thus for $L$-layer Transformer, the total FLOPs for $T'$-step autoregressive reasoning part is $\mathcal{O}(L\sum_{t'=1}^{T'}((|C|+|H|+t'-1)d+d^2))$.

\vspace{5pt}
\noindent
\textbf{ManCAR Overall.}
Combining these two parts, the total FLOPs for ManCAR is $\mathcal{O}(L((|C|+|H|)^2d+(|C|+|H|)d^2)+L\sum_{t'=1}^{T'}((|C|+|H|+t'-1)d+d^2))$.

\begin{table}[t]
\centering
\footnotesize
\caption{FLOPs for ERL, PRL, PLR, LARES, and ManCAR.}
\resizebox{0.85\linewidth}{!}{
\begin{tabular}{lc}
\hline
Method            & FLOPs \\ \hline
ERL/PRL & $\begin{aligned}
                      \mathcal{O}\Big(&L(|H|^2d+|H|d^2) \\
                      & +L\sum_{t'=1}^{T'}\big((|H|+t'-1)d+d^2\big)\Big)
                     \end{aligned}$ \\ \hline
PLR & $\begin{aligned}
                      \mathcal{O}\Big(&L(|H|^2d+|H|d^2) \\
                      & +L\sum_{t'=1}^{T'}\big((|H|+nt'-n)nd+nd^2\big)\Big)
                     \end{aligned}$ \\ \hline
LARES & $\begin{aligned}
                      \mathcal{O}\Big(&(L_{\mathrm{pre}}+L_{\mathrm{core}})(|H|^2d+|H|d^2) \\
                      & +L_{\mathrm{core}}{T'}\big(|H|d+d^2\big)\Big)
                     \end{aligned}$ \\ \hline
ManCAR            & $\begin{aligned}
                     \mathcal{O}\Big(&L\big((|C|+|H|)^2d+(|C|+|H|)d^2\big) \\
                     & +L\sum_{t'=1}^{T'}\big((|C|+|H|+t'-1)d+d^2\big)\Big)
                     \end{aligned}$ \\ \hline
\end{tabular}
}
\label{tab:flops}
\end{table}

Among latent reasoning baselines such as ERL, PRL, PLR, and LARES, the primary architectural difference in ManCAR lies in the introduction of the graph-conditioned context prompt $\mathcal{C}$, which extends the computation of sequential encoding. 
Specifically, for the recurrent reasoning method LARES, which consists of a $L_{\mathrm{pre}}$-layer pre-encoder and a $L_{\mathrm{core}}$-layer core encoder reused across reasoning iterations, the computational cost of the reasoning component scales as $L_{\mathrm{core}}\,T' \big(|H|d + d^2\big)$. For PLR, which adopts $n$ parallel reasoning streams, the cost of the reasoning component is $L \sum_{t'=1}^{T'} \big((|H| + n t' - n)nd + n d^2\big)$.

Tab.~\ref{tab:flops} summarizes the corresponding FLOPs for these methods. The additional computation in ManCAR mainly stems from processing this extra context, which has been shown in earlier experiments to yield substantial performance gains.
We argue that this overhead is justified, especially in light of the effectiveness of test-time scaling strategies widely adopted in modern LLM systems.

\end{document}